\documentclass[nonacm,acmsmall]{acmart}

\usepackage{preamble}

\begin{document}
\title{\sys{}: Herd Immunity against Storage Rollback Attacks in TEEs\tr{ [Technical Report]}{}}

\tr{
    \setcopyright{cc}
    \setcctype[4.0]{by}

    \settopmatter{printacmref=false}
}{
\setcopyright{cc}
\setcctype{by}
\acmJournal{PACMMOD}
\acmYear{2026} \acmVolume{4} \acmNumber{1 (SIGMOD)} \acmArticle{79} \acmMonth{2} \acmPrice{}\acmDOI{10.1145/3786693}
}



\author{David C. Y. Chu}
\orcid{0000-0001-9922-1994}
\affiliation{
    \institution{University of California, Berkeley}
    \country{USA}
}
\email{thedavidchu@berkeley.edu}

\author{Aditya Balasubramanian}
\orcid{0009-0005-8524-9490}
\affiliation{
    \institution{University of California, Berkeley}
    \country{USA}
}
\email{aditbala@berkeley.edu}

\author{Dee Bao}
\orcid{0009-0008-7592-6904}
\affiliation{
    \institution{University of California, Berkeley}
    \country{USA}
}
\email{dbao3@berkeley.edu}

\author{Natacha Crooks}
\orcid{0000-0002-3567-801X}
\affiliation{
    \institution{University of California, Berkeley}
    \country{USA}
}
\email{ncrooks@berkeley.edu}

\author{Heidi Howard}
\orcid{0000-0001-5256-7664}
\affiliation{
    \institution{Azure Research, Microsoft}
    \country{UK}
}
\email{heidi.howard@microsoft.com}

\author{Lucky E. Katahanas}
\orcid{0009-0008-3073-0844}
\affiliation{
    \institution{Unaffiliated}
    \country{USA}
}
\email{lkatahanas@gmail.com}

\author{Soujanya Ponnapalli}
\orcid{0009-0006-1449-1447}
\affiliation{
    \institution{University of California, Berkeley}
    \country{USA}
}
\email{soujanya@berkeley.edu}

\renewcommand{\shortauthors}{David C. Y. Chu et al.}

\begin{abstract}
\vspace{10pt}
\par Today, users can ``lift-and-shift'' unmodified applications into modern, VM-based Trusted Execution Environments (TEEs) in order to gain hardware-based security guarantees.
However, TEEs do not protect applications against disk rollback attacks, where persistent storage can be reverted to an earlier state after a crash; existing rollback resistance solutions either only support a subset of applications or require code modification.
% \souj{why is it fundamental is missing?}
Our key insight is that \emph{restoring disk consistency} after a rollback attack guarantees rollback resistance for \emph{any} application.
We present \textsc{Rollbaccine}, a device mapper that provides automatic rollback resistance for all applications by provably preserving disk consistency.
\textsc{Rollbaccine} intercepts and replicates writes to disk, restores lost state from backups during recovery, and minimizes overheads by taking advantage of the weak, multi-threaded semantics of disk operations.
% \changebars{Once the writes arrive at the replicas, \textsc{Rollbaccine} reconciles multi-threaded disk writes with single-threaded state machine replication by sending writes to disk in-order, while allowing non-conflicting writes to be concurrently in-flight.}{}
% \nc{Trying to save space/feels very detailed for an abstract}
% We introduce a formal specification of disk crash consistency and prove that it is preserved by  \textsc{Rollbaccine}. \david{Doesn't really flow well with this sentence, thought I could cut it out.}
% Across benchmarks over two real applications (PostgreSQL and HDFS) and two file systems (ext4 and xfs), \textsc{Rollbaccine} adds only $19\%$ overhead, except for the fsync-heavy Filebench Varmail.
% In addition, \textsc{Rollbaccine} outperforms the state-of-the-art, non-automatic rollback resistant solution by $208\times$.
\revisionbarsThree{\textsc{Rollbaccine} outperforms the state-of-the-art, non-automatic rollback resistant solution by $208\times$. In fact, }
{\textsc{Rollbaccine} performs on-par with state-of-the-art, non-automatic rollback resistant solutions; in fact, }
across benchmarks over PostgreSQL, HDFS, and two file systems (ext4 and xfs), \textsc{Rollbaccine} adds only $19\%$ overhead, except for the fsync-heavy Filebench Varmail.
% \heidi{Personally, I'd vote to reorder these last two sentences to really highlight the comparison to nimble. Maybe something like "Rollbaccine out performs the state-of-the-art.... In fact, across benchmarks over ...".}

% \souj{We could be more specific: Across 3 benchmarks over PostgreSQL and HDFS on two file systems (ext4 and xfs), \textsc{Rollbaccine} adds only $19\%$ overhead, except for xx\% for the fsync-heavy Filebench Varmail.}

\end{abstract}
\settopmatter{printfolios=true}
\maketitle
\pagestyle{plain}

\section{Introduction}
\label{sec:intro}

% Developers can lift-and-shift existing applications into VM-based Trusted Execution Environments (TEEs) today, gaining confidentiality and integrity guarantees with no code modification and minimal performance overhead~\cite{SNPPerfPost,SNPPerfPost2,cockroachCVM}.
Security-conscious developers lift-and-shift unmodified applications into VM-based Trusted Execution Environments (TEEs) under the impression that TEEs guarantee confidentiality and integrity with minimal performance overhead~\cite{SNPPerfPost,SNPPerfPost2,cockroachCVM}\footnote{It's important to highlight that TEEs \emph{only} protect confidentiality and integrity. An insecure application due to a code-level bug remains insecure.}.
This is true until the application needs to access disk; TEEs only protect data in memory, leaving the disk vulnerable.
% Unfortunately, this appealing narrative hides an ugly truth: TEEs only protect data in memory, leaving data on disk vulnerable.
% \nc{Do we want to directly say when nodes restart?}\heidi{I think not, today the data is vulnerable by default, its just fairly easy to handle whilst the cvm is online. Disk encryption is standard but not for integrity}
A combination of encryption, sealing, and hash verification can be used to provide confidentiality and integrity while the host is online, but once the host goes offline, the data on disk becomes vulnerable to \emph{rollback attacks}.

Rollback attacks revert disk to an earlier state, causing the system to execute over stale data.
Such attacks can be devastating: for example, an attacker can use rollback attacks in order to bypass limits on password attempts~\cite{androidRollbackKeys,iosRollbackKeys,signalRollbackKeys,juiceboxRollbackKeys} or reopen vulnerabilities in patched software~\cite{codeTransparencyService,windowsDowngradeAttack}. 
% To combat rollback attacks, existing TEE-based systems in production abandon the lift-and-shift philosophy in order to implement bespoke rollback protection mechanisms~\cite{svr3,ccf,gcpSpaces,azureConfidentialComputing,awsDmVerity}.

% Ideally, there would exist a solution against rollback attacks that is at once (1) \emph{general}, correct for all applications, (2) \emph{automatic}, requiring no application modification, and (3) \emph{resistant}, allowing the application to recover as if the rollback attack did not occur.
% Unfortunately, no existing solution provides all three guarantees.
% % ~\cite{rote,nimble,narrator,treaty,speicher,secureFS,engraft,svr3,vpfs,mlsdisk,dm-x,restartRollback,memoir,ariadne,ccf}
% Existing solutions either only detect but do not recover from rollbacks~\cite{treaty,speicher,secureFS,vpfs,mlsdisk,dm-x,memoir,ariadne}, are application specific~\cite{ccf,engraft,svr3,narrator,fastver}, or sacrifice automation by requiring the application to use a new API~\cite{nimble}.

% In this paper, we make one key observation: rollback attacks are fundamentally attacks on disks.
% As a result, generality can be achieved by restoring disk consistency after rollback, guaranteeing rollback resistance to \textit{any} application that uses disk regardless of application semantics.

To combat rollback attacks, production-level TEE-based systems implement \emph{bespoke} rollback-detecting or resistant solutions.
Signal's SVR3~\cite{svr3} modifies Raft~\cite{raft} while Azure's CCF~\cite{ccf} constructs a Merkle tree, whereas Google’s Confidential Space~\cite{gcpSpaces}, Azure’s Confidential Containers~\cite{azureConfidentialComputing}, and AWS’s Bottlerocket~\cite{awsDmVerity} verify code integrity at launch.
Academic efforts have further extended the catalog of supported applications~\cite{rote,nimble,narrator,treaty,speicher,secureFS,engraft,vpfs,mlsdisk,dm-x,restartRollback,memoir,ariadne}.

Lift-and-shift rollback resistance, however, remains elusive.
There is no solution that is at once (1) \emph{general}, correct for any application, (2) \emph{automatic}, requiring no application modification, and (3) \emph{rollback resistant}, allowing the application to recover as if the rollback attack did not occur.

We find in this paper that general, automatic rollback resistance is not only possible, but that \revisionbarsThree{it can also significantly outperform}{its performance is also competitive against} state-of-the-art, manual solutions.
Our key observation is that rollback attacks are fundamentally attacks on disk consistency.
General rollback resistance can therefore be achieved by restoring disk consistency after rollback, guaranteeing rollback resistance to \textit{any} application that uses disk regardless of application semantics.
% \souj{or file system} \david{I also think it's a bit too early to talk about file systems here (it's unclear why that's important at this point), we should just emphasize it once we discuss device mappers}
% \nc{I'm not sure the point is coming across yet. I know we don't want to talk about blocks just yet, but I worry that the jump from disk consistency to generality is non-obvious?}

The key challenge then lies in developing a strategy that preserves disk consistency at low cost.
Replication will necessarily be part of the solution: at least one machine must still have the data!
Na\"ively replicating all disk updates during execution, however, is a non-starter performance wise; this is why Nimble~\cite{nimble}, the state-of-the-art solution, requires the application to use a new API to indicate when replication is necessary.

A new API is unnecessary; this information is already available at the Linux block device layer~(\Cref{fig:linux-storage-stack}). Disk operations already include
% \souj{ when addressed at the block device layer? (the end result is nice but we could also emphasize our insight)}
% \nc{Broken transition with new text. I'd say: Disk operations already include persistence flags ...}
\textit{persistence flags} (\texttt{REQ\_FUA} and \texttt{REQ\_PREFLUSH})~\cite{fuaPreflush}, metadata attached to each write request indicating whether it should be synchronously written to disk or not.
Disk writes without these flags can return before persistence is guaranteed and potentially be lost after a crash.
These semantics are \emph{already} used by file systems in order to make most writes to disk fast and asynchronous by default, while a small number of operations are carefully persisted to ensure correctness.
% \souj{We build off these semantics (often under-specified and hidden within file system implementations) to replicate disk for rollback resistance} \david{Leaving out for now, seems like an unnecessary detail up front}
We can build off these semantics when replicating disk for rollback resistance; writes with persistence flags must be replicated on the critical path, while all other writes can be replicated in the background.

\begin{figure}[t]
    \centering
    \includegraphics[width=0.5\linewidth]{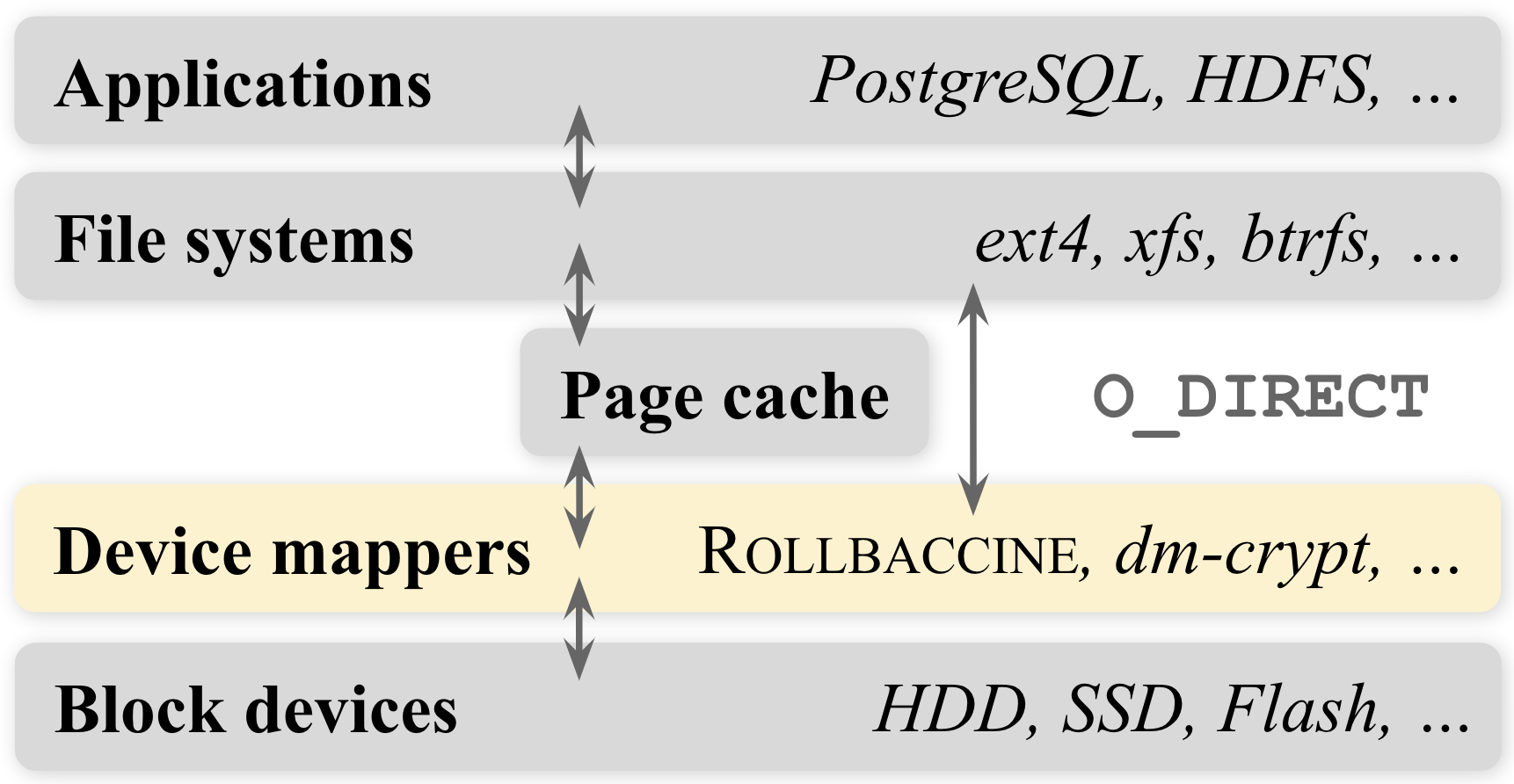}
    \caption{The Linux storage stack. Block I/Os tagged with \texttt{O\_DIRECT} bypass the page cache.
    % \souj{We should have \sys{} show up in this figure.}
    }
    \label{fig:linux-storage-stack}
\end{figure}

% Disk writes are notoriously slow, and applications and file systems \textit{already} operate under this assumption!
% As a consequence, file systems already make most writes to disk asynchronous by default, and instead very carefully choose a small number of operations to ensure persistence and overall correctness.
% These decisions are conveyed to the disk through \textit{persistence flags} (\texttt{REQ\_FUA} or \texttt{REQ\_PREFLUSH})~\cite{fuaPreflush}.
% These are the only writes that must be synchronously replicated to preserve disk consistency, and thus defend against rollbacks.

The weak semantics of disk also allow us to relax constraints on \emph{ordering}, further improving performance.
All existing countermeasures against rollback attacks enforce a total ordering of state changes in order to identify a ``canonical'' state that the system must recover to.
Disks are more flexible.
Upon crash and recovery, disks can recover any subset of weakly-persisted writes.
We can take advantage of this flexibility when replicating disk, allowing each disk to process writes in different orders and diverge, as long as they remain in a state consistent with prior operations.

We instantiate these ideas in \sys{} (the \textbf{rollba}ck va\textbf{ccine}), a system that intercepts and replicates writes to provide general and automatic rollback resistance with minimal overhead.
% \sys{} only replicates synchronous writes on the critical path, replicating asynchronous writes in the background.
% \sys{} identifies conflicting writes and enforces the order in which they are processed at each replica, but sends non-conflicting writes to disk in parallel.
% To prove that \sys{} restores disk consistency, we formally define the behavior of block devices (a category of storage devices that includes disk) in the presence of crashes (\Cref{sec:formalism}) and prove that block device consistency is preserved by \sys{} \tr{(\Cref{sec:correctness})}{in supplementary material}.
To prove that \sys{} restores disk consistency, we formally define the behavior of block devices (a category of storage devices that includes disk) in the presence of crashes and prove that block device consistency is preserved by \sys{}.
% \heidi{repetition between the above sentence and (2) in the contributions list, changebars show an possible cut down}

Importantly, we implement \sys{} as a device mapper \textit{below} the file system. This is key for providing generality: device mappers reason exclusively about block I/O requests and whether they should be written to disk synchronously or asynchronously. By preserving disk consistency at this level, \sys{} can defend against rollback attacks for any file system or application.

Our experimental results confirm that with \sys{}, general and automatic rollback resistance is possible with minimal performance penalty.
% On two large applications, PostgreSQL and HDFS, as well as two distinct file systems, ext4 and xfs, \sys{} introduces a maximum of $19\%$ throughput and latency overhead across TPC-C~\cite{tpc-c}, NNThroughputBenchmark~\cite{nnThroughputBenchmark}, and Filebench~\cite{filebench} Webserver, with more significant overheads ($71\%$ throughput and $2.7\times$ latency) only for the Filebench Varmail benchmark, with its high fsync frequency.
Across applications (PostgreSQL, HDFS, and the most common Linux file systems ext4 and xfs) and benchmarks (TPC-C~\cite{tpc-c}, NNThroughputBenchmark~\cite{nnThroughputBenchmark}, Filebench~\cite{filebench}) \sys{} introduces a maximum of $19\%$ throughput and latency overhead
\revisionbarsThree{, }{(on-par with Nimble~\cite{nimble}, a state-of-the-art, non-automatic rollback resistance solution)}
with more significant overheads ($71\%$ throughput and $2.7\times$ latency) only for Filebench Varmail, with its high fsync frequency.
% \souj{As discussed, I think the standard baseline should be encrypted local disk (no rollback protection at all) or a previous state-of-the-art for rollback protection?} \david{Keeping the baseline as the unreplicated disk for now, since we actually perform very well in comparison.}
% \nc{Highlight ext4 and xfs are the most common file systems?}
% In addition, \sys{} outperforms Nimble~\cite{nimble}---a state-of-the-art, non-automatic rollback resistance solution---by $208\times$ for write-heavy workloads.

In summary, we make the following contributions:
\begin{enumerate}
    \item We introduce \sys{}, a device mapper that offers applications rollback resistance (\Cref{sec:towards-rollbaccine}).
    \item We provide a formal definition of block device crash consistency (\Cref{sec:formalism}) and prove that it is preserved by \sys{} \tr{(\Cref{sec:correctness})}{in the supplementary material}.
    \item We show that \sys{} adds minimal overhead in most benchmarks and \revisionbarsThree{significantly outperforms}{is comparable to} state-of-the-art, non-automatic rollback-resistant solutions (\Cref{sec:evaluation}).
\end{enumerate}
\section{Motivation and Threat Model}
\label{sec:motivation}

TEEs provide confidentiality and integrity guarantees by
preventing, through hardware and software, unauthorized access to code or data.
Users can verify that their applications are executing within a TEE through remote attestation, where the host produces a proof of the code executing within a TEE~\cite{attestation}.Until recently, applications that wished to run within TEEs (e.g. Intel SGX) required extensive modifications and suffered significant performance penalties~\cite{scone}.

VM-based TEEs such as Intel TDX~\cite{tdx}, AMD SEV-SNP~\cite{SEV-SNP}, and Arm CCA~\cite{ARM-CCA} provide a new ``lift-and-shift'' abstraction, where applications can run \emph{unmodified} inside the TEE
% and gain confidentiality and integrity guarantees 
with minimal performance overhead.
TEEs have seen widespread adoption as a result: all major cloud providers support at least one type of VM-based TEE~\cite{azureConfidentialComputing,googleConfidentialComputing,awsConfidentialComputing}, and they are used in industry for private data processing~\cite{gcpSpaces,ccr,edgeless-mpc}, key management~\cite{signalRollbackKeys}, supply-chain security~\cite{codeTransparencyService}, and AI inference~\cite{azureConfidentialAI,nvidiaEdgeless,federatedLearning}.

\subsection{The Dangers of Rollbacks}
\label{sec:dangers-of-rollbacks}

Unfortunately, the confidentiality and integrity guarantees do not currently extend to persistent state.
% Attackers can observe and modify application data on disk by either directly accessing disk (using other applications) or by intercepting disk operations with privileged code (such as a kernel module loaded into the host OS).\heidi{Previous sentence could be cut in my opinion}
Existing encryption and integrity-preserving techniques~\cite{dmCrypt,dmIntegrity,bitlocker} can be used to automatically provide disk confidentiality and some degree of integrity.
However, because the metadata used to verify integrity is still stored on disk, they remain vulnerable to \emph{rollback attacks}, where an adversary could modify data (and its on-disk integrity metadata, in tandem) in order to present the application with a stale disk.

\begin{definition}[Rollback attack]
\label{def:rollback-attack}
Modifies disk reads such that they only reflect a prefix of prior operations.
\end{definition}

Online rollback attacks---performed while the application is executing---can be detected by an application that validates reads against integrity metadata in memory.
An attacker can instead launch an \emph{offline} rollback attack, crashing the host (and thereby clearing any metadata in memory) before rolling disk back.
Offline rollback attacks are insidious because they are \emph{undetectable}; the recovering application cannot distinguish an offline rollback attack from a benign crash, and will execute obliviously on stale state.

Consider for example a TEE application that rate-limits password-guessing attempts (\Cref{lst:password-guessing}).
It maintains a counter on disk to prevent excessive retries, even across reboots.
An attacker could repeatedly crash the TEE and rollback the disk to a state before the counter was incremented, effectively bypassing the guessing limit.%\souj{ with a rollback attack}. \david{Seems redundant, we already said "rollback the disk"}
% \souj{Clarify that for the purpose of this discussion the code is run inside TEEs}

\begin{figure}
\begin{lstlisting}[language=Python,
       keywordstyle=\color{blue}\ttfamily,
       stringstyle=\color{red}\ttfamily,
       commentstyle=\color{green}\ttfamily,
       xleftmargin=0.5\linewidth,
      ]
with open(counter_file, "r+") as file:
    counter = int(file.readline()) + 1
    if counter > 10: # Too many attempts
       return False
    else: 
        file.write(f"{counter}")
        file.flush()
        os.fsync(file.fileno())
        return pin == real_pin
\end{lstlisting}
    \caption{Password-guessing application}
    \label{lst:password-guessing}
\end{figure}

% \david{Comments that I removed from the program above for space:
% # Read the counter
% # Check whether exceeding max retries?
% # Write the counter}

% There are two main ways to launch a rollback attack.
% The first is by crashing the system, rolling the disk back, then rebooting the machine and pretending the crash was benign.
% \heidi{add sentence here to say that file system level intergity or dm-integrity does not prevent this as they store disks on disk.}
% Future application reads will fail to see updates that had been persisted to disk.
% The second is by intercepting disk operations, such as by loading a kernel module into the host OS and modifying the data returned on reads; although the disk is not physically modified, the disk state observed by the application is no longer consistent.
% \heidi{do we want to say something about rollbacking back whilst the host is online and why this isn't the focus of this paper}

Rollback attacks are relevant to \emph{all} applications that rely on persistent storage, including applications that do not interface with disk directly and instead rely on local (or even distributed!) database systems for persistence.
Those databases in turn rely on the persistence of \emph{their} local disks, which can be violated by rollback attacks.

% \changebars{To combat rollback attacks, production-level TEE-based systems all implement bespoke rollback-detecting or resistant solutions.
% SVR3~\cite{svr3}, Signal’s key recovery system, modifies Raft~\cite{raft} to prevent rollback attacks.
% CCF~\cite{ccf}, Azure’s TEE-protected ledger, constructs a Merkle tree to prevent rollbacks.
% Google’s Confidential Space~\cite{gcpSpaces}, Azure’s Confidential Containers~\cite{azureConfidentialComputing}, and AWS’s Bottlerocket~\cite{awsDmVerity} all verify the integrity of disk to detect rollbacks.}
% {}
% \david{Moved to intro.}\heidi{I preferred it here but only a very slight preference}
% \nc{What is a read-only partition?}
% \heidi{Not quite, I can only speak for ACI but whilst they do use dm-crypt and integrity (I need to double check that!) they don't yet protect againt rollback attacks. I'm hoping they will use rollbaccine for that :)}

\subsection{Threat Model and Guarantees}
\label{sec:threats}

\par \textbf{Threat model.} 
% We make the standard assumptions, common to all TEE-based systems~\cite{nimble,ccf,restartRollback,narrator,serverlessCC,cyclosa}, that clients trust the hardware manufacturer, and that TEEs are as safe as they claim to be.
We make the standard assumptions that clients trust the hardware manufacturer, and that TEEs are as safe as they claim to be~\cite{nimble,ccf,restartRollback,narrator,serverlessCC,cyclosa}.
Specifically, attackers cannot violate the integrity and confidentiality of memory, break standard cryptographic primitives, or exploit physical hardware or side-channel attacks~\cite{Gruss17, Oleksenko18, Van20, Van18, SGAxe, Borrello22, Murdock19, Skarlatos19, Schwarz17, VanSchaik22, asyncShock}.
\revisionbarsThree{}{Applications executed within a TEE will not deviate from their code.}
% \nc{Why not say: Applications executed within a TEE will always correctly execute code.} \david{I feel like ``correctly execute code'' might imply liveness (must execute all of its code) in addition to safety (won't execute random code), which isn't guaranteed.}
% Attackers, including malicious cloud providers, can still crash machines and corrupt network and disk I/O.
Attackers can still crash machines and corrupt network and disk I/O.

As in any system, the number of machines that an attacker can compromise directly impacts the correctness guarantees that the system can make; we therefore classify attackers as either \textbf{Type I (Individual)} or \textbf{Type II (Total)}.
Type I attackers represent malicious employees that can compromise up to $f$ machines, while Type II attackers represent malicious cloud providers with the ability to compromise all machines.
% \nc{I wonder if we can emphasise this is inherent to all systems: "As in any system, the number of machines that an attacker can compromise directly impact ...}
% \nc{I think we may want to give a little bit more context on why we're bothering to distinguish between the two? Right now it's a bit sudden. I think it'll feel stronger if we include some motivation}

\par \textbf{Correctness guarantees.}
Existing solutions provide one of two guarantees in the presence of rollbacks: rollback detection and rollback resistance.

\begin{definition}[Rollback detection]
\label{def:rollback-detecting}
An application is rollback-detecting if it detects rollback attacks and halts.
\end{definition}
% \nc{if it detects a rollback attack and halts? Halts sounds a bit negative}

Rollback detection guarantees \emph{safety}.
Following a rollback attack, a rollback-detecting application may be incapable of recovering, but it will never execute over stale data.
% In other words, rollback detection transforms rollback attacks into denial-of-service attacks.

\begin{definition}[Rollback resistance]
\label{def:rollback-resistant}
An application is rollback resistant if, following a rollback attack, it always recovers to a state it could have recovered to in the absence of rollback attacks.
\end{definition}

Rollback resistance guarantees \emph{liveness} in addition to safety.
Following a rollback attack, a rollback-resistant application will recover lost data and continue execution.

\sys{} guarantees rollback resistance in the presence of Type I attackers and rollback detection for Type II attackers.

\subsection{Limitations}
\label{sec:limitations}

% Rollback resistance guarantees that the effect of a rollback is invisible to the application.
% It does \emph{not}, however, guarantee correctness.
% % that the application recovers to a ``correct'' state or is secure against other forms of attacks.
% % Any vulnerability (aside from memory and disk confidentiality and integrity) that existed in the application remains present in the same rollback-resistant, TEE-protected application.
% % \natacha{Let's try to ask people to check they *get* this paragraph. Just to make sure we don't end up in the same situatin again}

% Consider again the password guessing application.
% If the application did not include the \texttt{os.fsync} line, an attacker could intelligently crash the application before the counter makes it to disk, thus bypassing the password guessing limit without ever modifying disk.
% Because the attacker did not rollback the disk, this is not a rollback attack and remains possible even if the application were rollback resistant.
% Similarly, rollback resistance's guarantees are limited to disk operations; it does not provide general safety to arbitrary applications, e.g. if the password counter relies on accurate system time from the untrusted host, or uses an insecure entropy source for password generation.

Type II attackers can always violate liveness (and rollback resistance) by simply denying service or shutting down all machines; this is a fundamental limitation of all rollback solutions ~\cite{nimble,secureKeeper,lcm,cyclosa}.
In that case, \sys{} falls back to rollback detection, ensuring that the application does not execute over tampered data.
For clients wary of cloud providers, a multi-cloud deployment would effectively convert all Type II attackers to Type I.

Importantly, rollback resistance should not be equated with application correctness.
An application with preexisting vulnerabilities, placed in a TEE and given rollback resistance, does not become secure.
Probabilistically unlikely failure modes in an application (or the underlying file system) may be exploited through non-rollback attacks.
% \heidi{In response to reviewer A from SOSP it might be worth adding a sentence here to say that we are also assuming that the file system is implemented correctly, ensuring this is an orthogonal problem addressed elsewhere (citing some papers of fs verification)}

Consider again the password-guessing application with the \texttt{os.fsync} line removed.
It is probabilistically correct in a benign setting where crashes are rare; however, an attacker could repeatedly crash the TEE before the counter makes it to disk in order to bypass the password guessing limit.
This is \emph{not} a rollback attack---the attacker crashed the TEE but did not modify the disk---and remains possible even with rollback resistance.
Rollback resistance simply removes the effect of rollback attacks.

\section{Towards \sys{}}
\label{sec:towards-rollbaccine}

The ideal solution for protecting against rollback attacks is general, automatic, and resistant, allowing developers to place unmodified applications in TEEs without worrying about rollback attacks. 

Generality---the ability to protect arbitrary applications against rollback attacks---is the hardest to achieve, because we do not know what data each application relies on for recovery.
To achieve generality, solutions like Nimble~\cite{nimble} sacrifice automation, requiring significant manual effort to identify the application-specific state that must be protected.

Our key insight is that nullifying the effect of rollback attacks on \emph{disk} is sufficient to guarantee general rollback resistance, regardless of individual applications' semantics.

By definition, regardless of how they are mounted, \textit{rollback attacks only modify disk}.
Thus, if we restore the disk to a state it could have recovered to after a benign crash, then to any application, the attack is indistinguishable from a benign crash.
The application must recover as if the rollback attack did not occur, granting it rollback resistance by definition (\ref{def:rollback-resistant}).

% Protecting disk state is challenging, as applications interact with the disk in complex ways.
% The choice of file system, for instance, influences the guarantees that an application can expect after a crash.
% File systems implement POSIX differently, and file system consistency is the subject of active research~\cite{ace,crashmonkey,alice,crashMonkeyAndACE,optimisticCrashConsistency,vijayThesis,daisyNFS,hydra,chipmunk,yggdrasil,ferrite,monarch}.

\sys{} leverages this insight to implement general, automatic rollback resistance by replicating disk.
Replication is a well-known strategy for recovering data lost in a rollback attack~\cite{nimble,narrator,rote,engraft}.
% \souj{Umm, do we want to ack that replication is a well-known strategy to undo rollbacks? First time we mention replication, I think?}
Replicating disk, however, requires addressing three main challenges:
(1) understanding exactly \emph{what} states the disk can recover to after a benign crash,
(2) intercepting and replicating writes to disk so they are available after a rollback attack, and
(3) limiting the overheads of doing so.
In the rest of this paper, we address each challenge in turn:

\par \textit{1. Formalisms (\Cref{sec:formalism}). } 
In order to restore disk to a state it could have recovered to after a benign crash, we need a formal understanding of exactly \emph{what} states it could have possibly been in.
% To the best of our knowledge, although \textit{file system} crash consistency~\cite{ace,crashmonkey,alice,crashMonkeyAndACE,optimisticCrashConsistency,daisyNFS,hydra,chipmunk,yggdrasil,ferrite,monarch,vijayThesis}
% has been extensively studied, there is no formal definition for the block-level semantics of disks.
Although file system crash consistency has been extensively studied~\cite{ace,crashmonkey,alice,crashMonkeyAndACE,optimisticCrashConsistency,daisyNFS,hydra,chipmunk,yggdrasil,ferrite,monarch,vijayThesis}, the field is primarily concerned with avoiding inconsistent states or detecting violations through testing.
To the best of our knowledge, no existing work formalizes the \emph{correct} behaviors for the disk that these file systems rely on.
%\souj{Although file system crash consistency has been extensively studied~\cite{ace,crashmonkey,alice,crashMonkeyAndACE,optimisticCrashConsistency,daisyNFS,hydra,chipmunk,yggdrasil,ferrite,monarch,vijayThesis}, the field is primarily concerned with avoiding inconsistent disk states by construction or detecting specific violations through testing. To the best of our knowledge, no existing work formalizes the set of block-level states that constitute correct disk executions—that is, \emph{all} the states a disk could recover to after a benign crash} \david{Reads longer but muddier to me. What's the detail you want to add?}
% \souj{The formal definition of block semantics is often hidden in the implementation of such file system operations.} \souj{Good place to discuss some "key" differences between fs crash consistency that changes with fs design vs block device crash consistency, which applies for all disks}

\begin{table}[t]
    \setlength\extrarowheight{-10pt}
    \footnotesize{
    \begin{tabular}{@{}ll}
        f = open("test.txt") & \texttt{WRITE(8)} \\
        & \texttt{READ(33928)} \\
        & \texttt{READ(1096)} \\
        write(f, "hello", 6) & \texttt{WRITE(1048664)} \\
        & \texttt{WRITE(1048672)} \\
        fsync(f) & \texttt{WRITE(1048680, FUA | PREFLUSH)} \\
        & \texttt{WRITE(1048680, FUA)} \\
        & \texttt{WRITE(266240)}
    \end{tabular}
    }
    \caption{File operations in ext4 and their corresponding block I/Os, sector numbers, and persistence flags.}
    \label{table:file-sys-calls-to-block-ios}
\end{table}

\par \textit{2. Intercepting (\Cref{sec:design}).  }
Next, we need a well defined interface for intercepting disk operations that is relatively simple to review, maintain, optimize, and trust, as opposed to a custom rollback-resistant application~\cite{nimble} or file system.
To this end, we implement \sys{} as a \textit{device-mapper}.

A Linux device mapper is a kernel module that lies between block devices (such as disks) and higher level applications, as seen in \Cref{fig:linux-storage-stack}.
Each disk read or write request arrives at the device mapper as a block I/O consisting of the disk sectors involved, pages containing data for writes or retrieving data for reads, and additional flags (\texttt{REQ\_FUA} and \texttt{REQ\_PREFLUSH}) describing whether the data should be written synchronously to disks or not.
\Cref{table:file-sys-calls-to-block-ios} describes a simple application writing to the \texttt{test.txt} file on the ext4 file system and the resulting block I/Os. 

Implementing \sys{} as a device mapper presents two benefits.
First, it is a commonly accepted strategy in industry to augment disk functionality.
Dm-crypt~\cite{dmCrypt}, dm-verity~\cite{dmVerity}, and dm-integrity~\cite{dmIntegrity} are all popular device mappers for enabling disk encryption and (limited) integrity without application modification~\cite{awsDmVerity,azureDmCrypt,azureDmVerity,azureStorage,androidDmCrypt,androidDmVerity}.
Second, device mappers sit below the file system (\Cref{fig:linux-storage-stack}) and are thus file-system agnostic; this allows us to evaluate against both ext4 and xfs in \Cref{sec:performance-overview} without code modification.
% Third, device mappers have a well defined interface that is relatively simple to review, maintain, optimize, and trust, as opposed to a custom rollback-resistant port of an existing application~\cite{nimble} or a custom rollback-resistant file system.

\par \textit{3. Overheads (\Cref{sec:design}).  }
Finally, we must minimize the overhead of replication.
Synchronous disk replication on the critical path not only introduces high overhead (\Cref{sec:nimble}), but is also often \emph{unnecessary}.
Existing applications and file systems already carefully engineer their implementation to reduce the number of synchronous writes; these writes are mapped to block I/Os with persistence flags
% (see the \texttt{fsync(f)} in \Cref{table:file-sys-calls-to-block-ios}) 
and exposed to the device mapper.
It suffices for \sys{} to synchronously replicate writes with persistence flags and replicate remaining writes in the background.

% \changebars{}{
% We provide a formal definition of block device semantics below and demonstrate that \sys{} preserves these semantics in \tr{\Cref{sec:correctness}}{supplementary material}.
% }
% \david{Commenting out for space}

% \souj{independent of a complex multi-threaded application or its choice of fs. And we prove this with the formal definition below?} \david{Wrong section to put it, since this is about overheads and when to sync/async replicate.}\souj{I guess the question is -- \textit{why} do we need to formalize block device semantics? it is the generality argument right -- so reiterating that here could help?}

% \souj{With block device crash consistency, we define \textit{all} safe behaviors/executions instead of characterizing a few unsafe behaviors like prior efforts\{cite: fs crash consistency and verification papers.\}; or we improve upon prior work to} \david{See edits to 1. Formalisms. I feel like our summary of the field right now misrepresents verification papers though, lmk how you'd like to fix}\souj{Done.}
\section{Block Device Crash Consistency}
\label{sec:formalism}

% \souj{To keep terminology consistent, shouldn't it be "safe state" not benign state?}

\sys{} provides rollback resistance by guaranteeing that, following a rollback attack, the disk always recovers to a 
% benign state.
% A formal definition of benign state is thus necessary.
% Concretely, we must formalize block device consistency in the presence of crashes: what are the states to which a block device (i.e. a disk) can recover to after a crash?
% These are the same states that \sys{} must recover to post-rollback.
safe state, one that it could have recovered to after a benign crash.
A formal definition of safe state, independent of file system semantics, is thus necessary.
Concretely, we formalize block device consistency as the states to which a block device (i.e. a disk) can recover to after a crash.
% Suggestion for rephrasing previous paragraph
% \souj{The disk states that a system can safely recover to post an adversary-crafted rollback are the same states that a disk recovers to after a benign crash. A formal definition of a benign state is thus necessary, however, independent of file system semantics. Concretely, we formalize block device crash consistency: the states to which a block device (i.e. a disk) can recover to after a crash.} \david{Adopted}

% \natacha{Unfortunately, to the best of our knowledge, there is no prior work that precisely describes block device crash consistency semantics. 
% The documentation of their behavior is instead dispersed throughout Linux mailing lists~\cite{dmIntegrityMailingList} and the assumptions made by various file system consistency papers~\cite{crashmonkey,crashMonkeyAndACE}.
% This is unlike file system crash consistency, which is an active area of research~\cite{vijayThesis}.
% We therefore begin by formalizing block device crash consistency.}
% \david{I already included this paragraph in the previous section under formalisms; I vote we keep it there (since it's more related-work flavored) and avoid repetition.}\natacha{yup makes sense}

% Crash
% full system crash (all processes/threads) (all volatile state lost)
% We start by defining the behavior of block devices (general storage systems that includes disks) in the event of a crash ($C$), where the entire system---all threads and their volatile state---is reset.
We begin with the assumption that disk read and write \textit{operations} ($O$) are atomic, which is consistent with prior work and the inherent properties of disks~\cite{brown2002oracle, sivathanu2003semantically, pillai2014all}.
We also assume that writes to block devices are guaranteed to persist across crashes only when the persistence flags \texttt{REQ\_FUA} or \texttt{REQ\_PREFLUSH} are used~\cite{crashmonkey,crashMonkeyAndACE,vijayThesis,fuaPreflush}.
% \david{Changed wording here (used to be "only when they are tagged with the persistence flags) because REQ_PREFLUSH gives guarantees to non-tagged writes as well.}
\texttt{REQ\_FUA} guarantees that when the write is completed, it must be persisted, whereas \texttt{REQ\_PREFLUSH} guarantees that any previously completed operation is persisted.

We describe the execution of a block device with a \textit{history} $\mathcal{H}$ as a totally ordered sequence of events $V$ composed of \textit{invocations}, \textit{responses}, and \textit{crashes}.
This total ordering allows us to capture causal relationships between reads and writes and is distinct from the (unknown) order in which the disk actually processes operations.
We denote $\mathcal{H}[t]$ as the sequence of events performed by a thread $t$.
% \nc{Can we justify that the history is totally ordered?}

% Invocations and responses are associated with read or write \textit{operations} $O$. 
Concretely, invocations represent \texttt{submit\_bio} function calls sent to the block device; responses represent the corresponding call to \texttt{bi\_end\_io} by the block device, signaling I/O completion. 

Read requests to block $b$ by thread $t$ are written $R_{inv,t}(b)$; 
responses are $R_{res,t}(b,val)$, where $val$ is the value returned.
Write requests are $W_{inv,t}(b,val,sync)$, where the value to write (if any) is \textit{val} and \textit{sync} is a tag with one of the following values: \texttt{REQ\_FUA}, \texttt{REQ\_PREFLUSH}, \texttt{REQ\_FUA|REQ\_PREFLUSH}, or $\varnothing$.
$W_{res,t}(b)$ is the matching response.
% \nc{nit: does bi\_end\_io not include an error code?} \david{Yes, but in our formalism we assume that all operations succeed. Do you think we should include it or say we ignore it?} \nc{Prob ok to ignore}
We assume that blocks are always written to before they are read from.
% \souj{I like that persistence flags are explained in line -- maybe just adding a line saying they are explained in line could help?}

To define block device crash consistency, we take as our starting point the definitions of Izraelevitz et al \cite{durableLinearizability}.
% Our definitions differ due to semantic differences between the persistence operations of NVM (\texttt{pwb}, \texttt{pfence}, \texttt{psync}) and block devices (\texttt{REQ\_FUA} and \texttt{REQ\_PREFLUSH}).
% This formalization requires reasoning about the observable state on the block device in the presence of crashes and concurrent, multi-threaded executions.
We will build up to a definition of linearizability before extending it to crashes.
We start by defining a sequential history.
% \changebars{We then make explicit which pieces of state must be durable after a crash with durable cuts and define block device crash consistency.}{} \nc{Not sure that paragraph was particularly helpful for me. In the interest of needing space, cut?}

In a sequential history, responses always follow invocations.
There can be at most one \textit{pending} invocation at a time (invocation without a matching response) and a crash cannot occur between an invocation and its response.

\begin{definition}[Sequential history]
\label{def:sequential}
A history $\mathcal{H}$ is sequential if for each $O_{inv}$ and its matching response $O_{res}$ in $\mathcal{H}$, $\exists \mathcal{H}_1, \mathcal{H}_2$ such that $\mathcal{H} = \mathcal{H}_1O_{inv}O_{res}\mathcal{H}_2$.
\end{definition}

This allows us to reason about multi-threaded histories by comparing each thread's execution to a sequential history.\footnote{For programs that issue concurrent operations per thread using async I/O, we can map each physical thread to multiple abstract threads.}
When multiple threads operate over the same block, we use the \textit{happens-before} relationship to order writes and reads on different threads, as this is necessary to determine whether a history satisfies reads-see-writes.

\begin{definition}[Happens-before]
\label{def:happens-before}
An event $V_1$ happens-before event $V_2$ in a history $\mathcal{H}$ (denoted $V_1 \prec V_2$) if $V_1$ precedes $V_2$ and either \\
(1) $V_1 = O_{res}(b)$ and $V_2 = O'_{inv}(b)$ over the same block $b$, \\
(2) $V_1$ or $V_2$ is a crash $C$, \\
(3) $V_1 = O_{res}(b)$ and $V_2 = W_{inv}(b',val,sync)$ where $sync$ contains \texttt{REQ\_PREFLUSH}, or \\
(4) $\exists V'$ such that $V_1 \prec V' \prec V_2$.
\end{definition}

Criteria 1 and 4 are standard~\cite{linearizability}; 
% Criterion 2 was introduced for crash-consistent NVMs~\cite{durableLinearizability} and also applies to crash-consistent block devices.
% It states that crashes are global events; all events either happen-before or after a crash.
Criterion 2 states that crashes are global events; all events either happen-before or after a crash~\cite{durableLinearizability}.
Criterion 3 is new and captures the global semantics of \texttt{REQ\_PREFLUSH}: once a \texttt{REQ\_PREFLUSH} is invoked, 
% it is only returned by disk once all previous writes from every thread are flushed and persisted, regardless of which blocks they wrote to.
the operation only completes when all previous writes from every thread are flushed and persisted.
% Applying this criteria to \Cref{fig:formalism}, $W_{res}(1) \prec W_{inv}(0,\texttt{REQ\_PREFLUSH})$, but $W_{inv}(2) \not\prec W_{inv}(0,\texttt{REQ\_PREFLUSH})$.
% Once the \texttt{REQ\_PREFLUSH} completes, we are guaranteed that $W(1)$ is persisted to disk, but have no guarantees on $W(2)$.
% Note that criteria 3 does not apply to \texttt{REQ\_FUA}, which only guarantees that the tagged write itself persists to disk~\cite{fuaPreflush}. 
% This will become relevant once we discuss the state a block device recovers to after a crash. \nc{This last sentence can optionally be cut if we need space}

The happens-before relationship allows us to define what each read returns.
Each read of block $b$ must return the value of the latest completed write to that same block $b$, as long as there are no crashes in-between (during which writes may be lost). 
We formalize this as the \textit{reads-see-writes} property, which is only defined for crash-free periods of history (we call these \textit{era}s $\mathcal{E}$).
% \nc{See conversation with David about potentially simplifying notation}.

\begin{definition}[Reads-see-writes]
\label{def:reads-see-writes}
A history $\mathcal{H}$ respects reads-see-writes if $\forall R_{res}(b,val) \in \mathcal{H}$, there is a preceding write invocation $W_{inv}(b,val,sync)$ with that same $val$ such that $\mathcal{H} = \mathcal{H}_0W_{inv}\mathcal{E}_0W_{res}\mathcal{E}_1R_{inv}\mathcal{E}_2R_{res}\mathcal{H}_1$, and there does not exist another $W_{inv}(b)$ in the eras $\mathcal{E}_0, \mathcal{E}_1, \mathcal{E}_2$.
\end{definition}

%For now, we can model multi-threaded histories as concurrency-free, sequential histories using \textit{linearizability}.

Finally, we  consider \textit{pending} invocations: invocations without a matching response.
Pending writes in particular require care as they may (or may not) have been processed by the underlying block device and reflected in the next read.
We write \textit{compl}$(\mathcal{H})$ to be the set of histories generated from $\mathcal{H}$ by inserting matching responses after some pending invocations.
This models situations where pending operations \textit{have} been persisted to disk.
In contrast, let \textit{trunc}$(\mathcal{H})$ be the history generated from $\mathcal{H}$ by removing all pending invocations.
This reflects histories where the operation was \textit{not} persisted.

% \changebars{Intuitively, a well-formed history is \textit{linearizable} if it is equivalent to a sequential history that respects happens-before relationships across threads.
% Formally, a well-formed history $\mathcal{H}$ and sequential history $\mathcal{S}$ are \textit{equivalent} if $\forall t$, $\mathcal{H}[t] = \mathcal{S}[t]$.
% }{}

% \changebars{Reasoning about invocations without matching responses requires a bit more care: either the invocation has been processed by the underlying block device, or it has not.
% Let \textit{compl}$(\mathcal{H})$ be the set of well-formed histories generated from $\mathcal{H}$ by inserting matching responses after some pending invocations.
% This models the fact that pending operations may have been persisted to disk.
% Let \textit{trunc}$(\mathcal{H})$ be the history generated from $\mathcal{H}$ by removing all pending invocations.
% }{}

We can now define \textit{linearizable history} as follows.

\begin{definition}[Linearizable history]
\label{def:linearizable-history}
A history $\mathcal{H}$ is linearizable if there exists a history $\mathcal{H}' \in \textit{trunc}(\textit{compl}(\mathcal{H}))$ and a sequential history $\mathcal{S}$ such that: \\
(1) $\mathcal{S}$ respects reads-see-writes \\
(2) $\forall t$, $\mathcal{H}'[t] = \mathcal{S}[t]$ (i.e. $\mathcal{H}'$ and $\mathcal{S}$ are equivalent)\\
% $\mathcal{H}'$ and $\mathcal{S}$ are equivalent
(3) $V_1 \prec V_2$ in $\mathcal{H}'$ implies $V_1 \prec V_2$ in $\mathcal{S}$.
\end{definition}

% In \Cref{fig:formalism}, $\mathcal{H}$ is linearizable if $R_{res}(1)$ the value written by $W_{inv}(1)$, respecting reads-see-writes.
% Then, $\mathcal{H}$ would be equivalent to the following sequential history: $W_{inv}(1)W_{res}(1)W_{inv}(0)W_{res}(0)R_{inv}(1)R_{res}(1)$.

% Since $W_{inv}(2)$ is a pending invocation, we can either remove it from the history with \textit{trunc} or append a matching response after it with \textit{compl}.
% \changebars{If the read response for thread 1 returns any other value, however, no equivalent sequential history would exist for $\mathcal{H}$ that also respects reads-see-writes, so $\mathcal{H}$ would not be linearizable.}{} \nc{I'm not sure that the second half was a particularly useful, I think it'd be useful if there was another example that looked all good until you "tried" to construct a sequential history with it}
% \heidi{readabiliy of this would be helped with an example, either as a figure or in the text where (a) is a linearizable history and (b) is not a linearizable history}
% \nc{Definitely agree with Heidi. These definitions need to be seen as having value to a systems audience. I'd recommend trying to see what we did in Basil, we tried really hard to provide intuition for them.}

In the absence of crashes, this definition of linearizability is sufficient to model the behavior of multi-threaded operations over a block device.
With crashes on the other hand, some but not all writes may be recoverable from the block device.
We formalize the set of possible write values that may be recoverable with the notion \textit{durable cut}.

\begin{definition}[Durable cut]
\label{def:durable-cut}
A durable cut $\mathcal{D}$ of history $\mathcal{H}$ is a subhistory of some $\mathcal{H}' \in \textit{trunc}(\textit{compl}(\mathcal{H}))$ where \\
(1) if $\mathcal{H}'$ contains $W_{inv}(b,val,sync)$ and its matching response $W_{res}(b)$ where $sync$ contains \texttt{REQ\_FUA} or \texttt{REQ\_PREFLUSH}, then $\mathcal{D}$ must contain $W_{res}(b)$, \\
(2) $\forall V \in \mathcal{D}$, $\mathcal{D}$ also contains any $V'$ where $V' \prec V$ in $\mathcal{H}'$, and \\
(3) $\mathcal{D}$ has no pending invocations.
\end{definition}

The durable cut is a \emph{cut} of history that contains (1) all writes tagged with persistence flags and (2) any writes that happen-before a write already in the cut.
This is where the extra criteria for \texttt{REQ\_PREFLUSH} in the happens-before relationship becomes relevant; if a \texttt{REQ\_PREFLUSH} completed, then it must be in the durable cut, and any operations that happened-before it must be in the cut as well.
% The \texttt{REQ\_FUA} flag instead simply guarantees that, if an operation completes with that flag, it will necessarily be in the cut.

% For example, if $\mathcal{H}$ in \Cref{fig:formalism} is followed by a crash, then the durable cut must contain $W_{inv}(0)$ and its matching response, since it is a \texttt{REQ\_PREFLUSH}, and $W_{inv}(1)$ and its matching response, since it happens-before $W_{inv}(0)$.
% The durable cut need not contain $W_{inv}(2)$ or $R_{inv}(1)$ and its matching response because they do not happen-before $W_{inv}(0)$.

Finally, we can formalize what persisted state can be read from disk after a crash with \textit{block device crash consistency}.
Effectively, for each crash-free era $\mathcal{E}$, the set of writes that ``made it to disk'' before the crash forms the durable cut $\mathcal{D}$.
% , where any completed writes tagged with persistence flags must have ``made it to disk''. \david{Again, removed because PREFLUSH provides guarantees for non-tagged writes as well}

\begin{definition}[Block device crash consistency]
\label{def:block-device-crash-consistency}
A history $\mathcal{H} = \mathcal{E}_0C_0\mathcal{E}_1C_1\ldots\mathcal{E}_{x-1}C_{x-1}\mathcal{E}_x$ is block device crash consistent if there exists a single $\mathcal{D} = \mathcal{D}_0\mathcal{D}_1\ldots\mathcal{D}_{x-1}$ such that $\forall i$, $\mathcal{D}_i$ is a durable cut of each era $\mathcal{E}_i$, and $\mathcal{D}_0\mathcal{D}_1\ldots\mathcal{D}_{i-1}\mathcal{E}_i$ is linearizable.
\end{definition}
% \heidi{again I think this definition would benefit from two examples }
% \nc{Maybe also create an init event that initialises the history to $\bot$? Does H need to be linearizable, or just DoD1D2..} \david{Just D0D1D2. H can violate reads-see-writes because writes can be lost after a crash}

Block device crash consistency checks the following.
Is era $\mathcal{E}_0$ linearizable?
Then, moving on to $\mathcal{E}_1$, is there some durable cut $\mathcal{D}_0$ of $\mathcal{E}_0$ (representing the writes that had actually made it to disk) such that $\mathcal{D}_0\mathcal{E}_1$ is linearizable?
It builds inductively, keeping the data that made it to disk consistent for each era.
If the above holds for \emph{all} eras in $\mathcal{H}$, then $\mathcal{H}$ is block device crash consistent.
% \nc{Had a hard time following this paragraph, a concrete example would really help. It's non standard to have to build an equivalence definition iteratively, which is I think why i'm struggling}
% \nc{Still hard to follow the inductive approach. Fine for now, but we probably want a slightly better way to present it later}

% In \Cref{fig:formalism}, $\mathcal{H}$ is a linearizable era, satisfying the condition for $\mathcal{E}_0$.
% If it is followed by a crash, then it remains block device crash consistent if $\mathcal{E}_1$, the era after the crash, satisfies the following:
% any $R_{res}(1)$ reflect the value written by $W_{inv}(1)$, since $W_{inv}(1)$ and $W_{res}(1)$ must be in the durable cut, and
% $R_{res}(2)$ do or do not see $W_{inv}(2)$.
% If $R_{res}(2)$ \emph{does} see $W_{inv}(2)$, then the inductive property of block device crash consistency states that even after \emph{another} crash, the next read will \emph{still} see $W_{inv}(2)$, since the same durable cut that contains $W_{inv}(2)$ is reused for each inductive step.

% \souj{Comes at an odd place; merge it with the standard assumption above?}
Assuming only benign system crashes and no random disk corruption, block device crash consistency precisely captures the set of histories produced by a disk~\cite{dmIntegrityMailingList}.
% We prove that all histories produced by \sys{} are block device crash consistent in \tr{\Cref{sec:correctness}}{the supplementary material}.
% In this way, \sys{} ensures that the system always remains in a benign state, thus guaranteeing rollback resistance. 
% This guarantee is file system- and application-agnostic.
We prove that all histories produced by \sys{} are block device crash consistent in \tr{\Cref{sec:correctness}}{the supplementary material}, ensuring that the system always remains in a safe state and is thus rollback resistant.
This guarantee is file system and application-agnostic.

% \natacha{For instance,  any crash-consistent file system \textit{must remain consistent} over \textit{any} block device that produces block device crash consistent histories~\cite{ace,crashmonkey,crashMonkeyAndACE}. As such, \sys{} guarantees that any file system atop \sys{} remains crash-consistent atop rollbacks. Transitively, any database that is correctly implemented atop a crash-consistent file system}

\section{System Model}
\label{sec:system-model}

\sys{} maintains block device crash consistency in the presence of rollback attacks through fault-tolerant replication of disk writes.

\sys{} consists of $N$ machines, running within TEEs.
One machine is the \textit{primary} where the application executes, while the remaining $N-1$ nodes are \textit{backups}.
During execution, the primary replicates writes to at least $f$ backups.
After a crash (and potential rollback), the recovering machine's disk is restored to a block device crash-consistent state by contacting at least $N-f$ existing machines and recovering from the most up-to-date machine.
% \heidi{why $f$ for recovery? I was expecting $N-f$}

\par \textbf{Selecting $N$.}
$N$ can be configured to be any value between $f+1$ and $2f+1$ (the largest $N$ where the replication and recovery quorums still intersect).
The relationship between $N$ and $f$ represent a configurable tradeoff between availability and cost.
Traditional consensus protocols maximize availability by setting $N=2f+1$, reducing the effect of failures on recovery\revisionbarsThree{}{ (remaining live with up to $f$ failures)} at the cost of additional machines.
Primary-backup and chain replication~\cite{chainReplication} protocols minimize cost by setting $N=f+1$, reducing the number of machines at the cost of availability when nodes fail.

By default, \sys{} sets $N = f+1$ and requires explicit recovery to recoup liveness; \sys{} is only live in the absence of failures.
This is in line with what a user can expect from a traditional cloud deployment: if a VM crashes, the user (or some third-party software) is responsible for restarting it or deploying another VM.
\revisionbarsThree{}{Setting $N=2f+1$ guarantees \sys{}'s liveness with up to $f$ backup failures.}

% \david{Cutting correctness guarantees completely here since we already go over it in section 2?}

\section{Design}
\label{sec:design}

\begin{figure}
\begin{subfigure}[t]{0.60\linewidth}
    \centering
    \includegraphics[width=\textwidth]{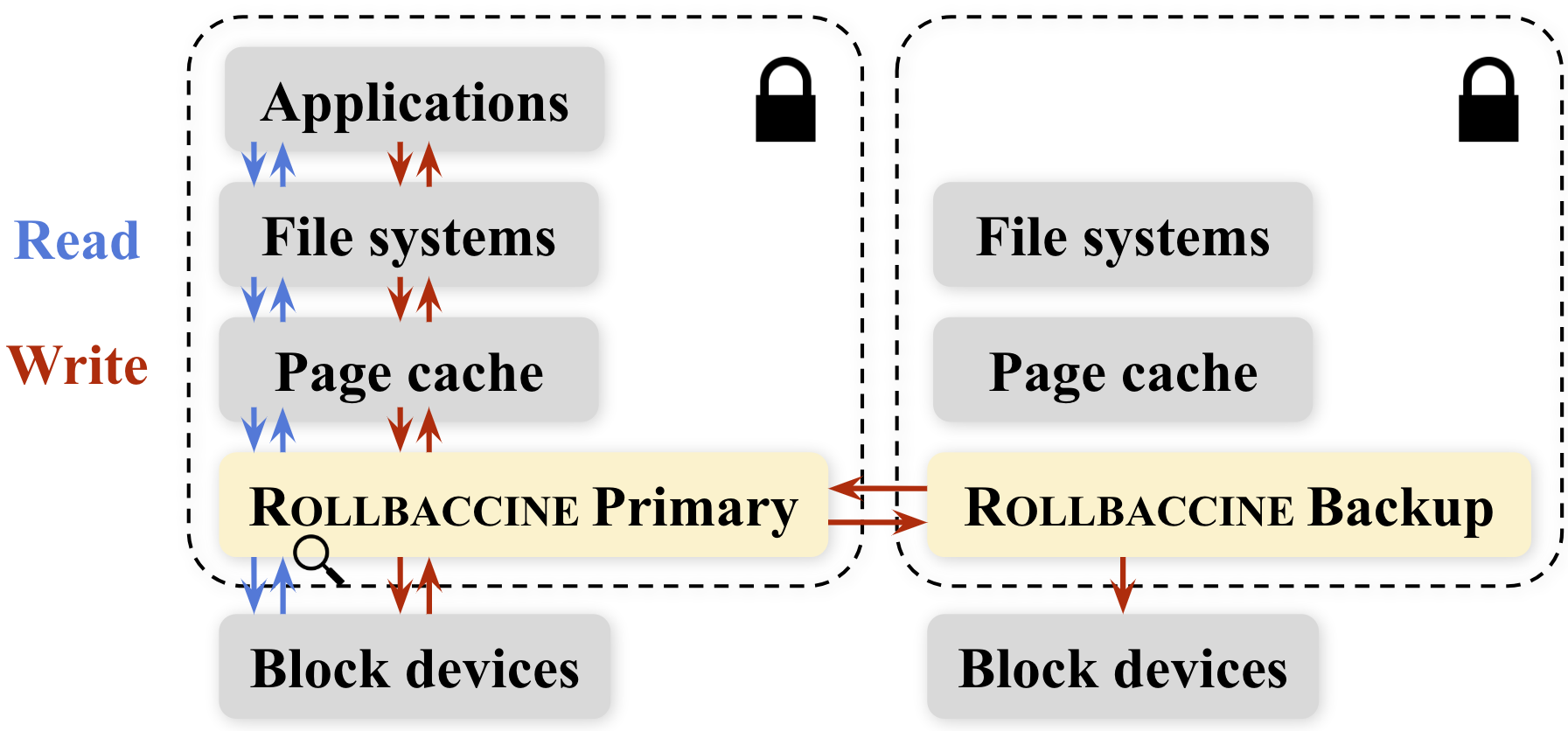}
    \caption{\sys{} on the critical path with $f = 1$.}
    \label{fig:architecture}
\end{subfigure}
\begin{subfigure}[t]{0.39\linewidth}
    \centering
    \includegraphics[width=\textwidth]{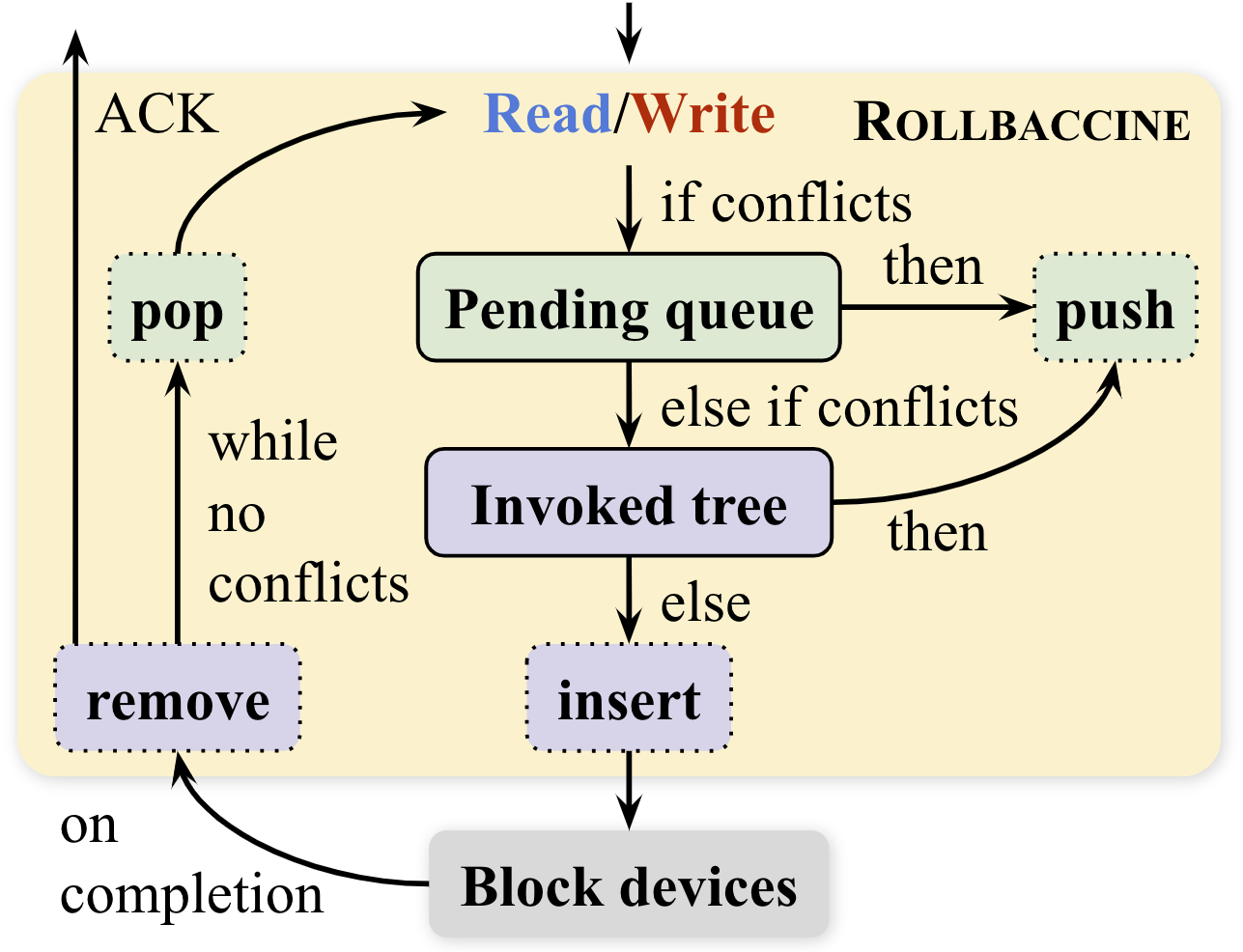}
    \caption{\sys{} concurrency handling. ``Pop'' and ``push'' are operations over the pending queue, and ``insert'' and ``remove'' are operations over the invoked tree.}
    \label{fig:pending-queue-invoked-tree}
\end{subfigure}
\caption{\sys{}'s design.}
\end{figure}

% \nc{I also think that it might be nice to talk about the tradeoffs between replicating for rollback protection and rollback detection? For instance, by hving the figure include the data and the digests? Where if you only replicate the digests then you have rollback detection only, if you replicate data as well, you also have rollback protection. } 
\sys{} must balance the high cost of replication with the constraints placed by block devices.
This tension manifests itself in two areas: synchronous vs asynchronous replication, and multi-threaded vs single-threaded execution.

\textbf{Avoiding unnecessary synchronous replication.}
Na\"ive, synchronous replication of all disk writes to a set of backups is prohibitively expensive, as it requires waiting for responses from sufficiently many backups before acknowledging each write.
Asynchronous replication, on the other hand, introduces a window of vulnerability during which data may be lost: the write may have optimistically been confirmed before replicating to sufficiently many backups.

\sys{} recognizes that applications and file systems already trade-off between performance and persistence: writes are asynchronous by default unless synchronized through operations like \texttt{fsync} or flags like \texttt{O\_SYNC}.
% \texttt{Open("hi.txt",O\_SYNC)}, for example, guarantees that writes to the file will not return before data has been transferred to the disk.
It is already the case that, if the system crashes, the disk is under no obligation to persist asynchronous writes.
\sys{} needs only to provide this same guarantee.
\sys{} thus only synchronously replicates writes tagged with persistence flags and asynchronously backs up all other writes.
% \nc{I don't think this is the right place to talk about the block layer. The block layer arrives (when you list the challenges) in the intro, as a way to handle parallelism without having to worry about whether system calls persist. I'd move over the "We do it at the block layer" in the second section}

\textbf{Multi-threaded execution for non-conflicting operations.}
Disks achieve high throughput by allowing writes to be processed in parallel.
% On the other hand, replicated state machines typically require operations to be replicated in the same order across all members to ensure that states converge.
% In practice, this usually means that operations are applied by a single thread.
% \natacha{BlockSTM, Eve, all use multithreading to apply operations. Would significantly soften claim}
% \heidi{CBs propose softening this sentence, reviews might object to too broader claim here}
To maintain the multithreaded nature of disks when replicating, \sys{} exploits the fact that write invocations between replicas \emph{need not be processed in the same order}.
Because the backups' states are only used in the event of a crash or rollback attack, they simply need to be durable cuts (\Cref{def:durable-cut}) of the primary's state in order to achieve block device crash consistency.
In other words, the backups must respect happens-before relationships and the semantics of persistence flags, but are free to reorder all other operations.
Concretely, \sys{} backups submit write block I/Os to disk according to the total order assigned by the primary, but submissions do not block on the completion of previous I/Os unless they conflict.

As a result, states may actually \emph{diverge} between backups.
However, since all backups maintain a durable cut, the primary can correctly recover from \emph{any} backup's state and still maintain block device crash consistency by definition (\ref{def:block-device-crash-consistency}).
% \natacha{the blocking comment isnt very clear} \david{Was trying to address the concern that after recovery, backups with different durable cuts will diverge in an inconsistent way. Will think about how to reword.}
% \heidi{Isn't it just any backup since the paper currently requires us to wait for all f backups? Picking the freshest is just a optimisation to preserve as many writes as possible but not needed for safety.} \david{Commented out since now we have $N$ and there may be less fresh backups}

% \souj{I find it interesting that the write-persistence order could diverge between primary-backup and every pair of backups too :p -- but it is unclear what relaxes means here - not sequential but?}

% In this section, we first describe \sys{}'s critical path (\Cref{sec:critical-path}), then initialization and recovery (\Cref{sec:initialization-and-recovery}).

\subsection{Critical Path}
\label{sec:critical-path}

% \nc{Is there any chance we could get a diagram that actually maps down a open() write() fsync() to a sequence of block IOs and flags? And then use that as a running example? Otherwise it's still pretty abstract}
We first discuss the steady-state of \sys{}.

\subsubsection{Asynchronous writes on the primary}
\label{sec:async-writes-primary}
Writes without persistence flags, such as those before the fsync in \Cref{table:file-sys-calls-to-block-ios}, are asynchronously replicated to backups.
When \sys{} intercepts a write on the primary $p$, it encrypts and hashes it (with authenticated encryption), stores the hash in memory, then atomically (1) assigns it a monotonically increasing \texttt{writeIndex}$_p$ and (2) places it on the network queue, where it will be signed and sent to the backups.
% \natacha{what do you mean with authenticated encryption? Strange phrasing} \david{I mean AEAD (Authenticated Encryption with Additional Data), does both encryption and integrity checking}

\textbf{Keeping integrity metadata in memory.}
Traditional integrity-preserving systems that keep integrity metadata on disk~\cite{dmIntegrity,azureStorage} are vulnerable to attacks that simultaneously rollback the data and its integrity metadata.

\sys{} instead replicates integrity metadata in-memory, relying on the TEE's integrity guarantees while the machine is online; once offline, integrity metadata must be retrieved from backups during recovery.
To reduce \sys{}'s memory footprint, we create a Merkle tree of hashes and store the lower $L$ layers on disk, verifying any hashes read from disk against the higher layers.
The configuration of $L$ represents a tradeoff between memory usage and read/write amplification from accessing additional blocks on disk.

% \david{Leaving out for now, don't think anyone actually misses this?}
% The total size of the integrity metadata is $\frac{\textit{disk size} \times \textit{hash}}{\textit{granularity}}$, where \textit{granularity} is the size of each block in a block I/O (4KB) and \textit{hash} is the size of the hash of each encrypted block (16B for AES-GCM), totalling $0.4\%$ the size of disk.
% Put in context, 600GB of disk requires 2.4GB of metadata.

Prior work has also explored using Merkle trees (without replication) to detect disk integrity violations~\cite{dm-x,vpfs,mlsdisk,memoir,ariadne,fastver}.
Their correctness rests on keeping the root/tail hash in ``small trusted storage''.
Even if small trusted storage were available (and evidence suggests otherwise~\cite{rote,nimble}), these solutions are at best rollback \emph{detecting} and not \emph{resistant}; once the metadata is corrupted, it cannot be recovered.

\textbf{Managing the integrity of concurrent conflicting writes.}
% After storing its hash \natacha{awk sentence. First question that comes to mind is "storing where?"}, \sys{} submits the write to disk, returning to the upper layer \natacha{what is the upper layer? } when the disk acknowledges the write.
\sys{} then submits the encrypted write to disk, signaling completion once it is acknowledged by disk.

Unfortunately, submitting writes to disk without blocking on previous writes' completion complicates the maintenance of integrity metadata.
Consider two concurrent writes $W,W'$ to block $b$ where $W_{inv}(b) \prec W'_{inv}(b) \prec W_{res}(b)$.
The integrity metadata must match the data of the ``later'' write, but the concurrency prevents us from knowing which write was last.
% In order to verify the integrity of a later read over the same block \nc{Given that you haven't mentioned reads yet, this transition is confusing}, the primary must know which write was processed last, but it cannot know because the writes are concurrent.
% \nc{Why does it matter if you don't know the ordering. Could you not keep both if you had a multiversioned system?. What does it actually mean for a write to be concurrent?} \david{I'm not sure how to really answer this question, a multiversioned system could blow up infinitely (if you had 100 concurrent writes) and also be difficult to recover to consistently. Commenting out for now}

To address this issue, we impose an ordering on same-block writes by maintaining two data structures: a tree of \emph{invoked} writes, sorted by write location, and a queue of \emph{pending} writes, seen in \Cref{fig:pending-queue-invoked-tree}.
After assigning each write a \texttt{writeIndex}$_p$, the primary atomically checks if it conflicts with any other invoked or pending write.
If it does, then the write is placed on the pending write queue and waits to be unblocked.
Otherwise, the primary stores its hash, adds the write to the invoked write tree, and submits it to disk.
Once the write completes, it is removed from the invoked write tree, and any non-conflicting writes are popped off the pending queue in-order and submitted to disk.
At this point, the asynchronous write is marked completed.
% \david{Diagram?} \heidi{yes please} \nc{Agreed. Also having a single description of everything at the top, with the key challenge being how to handle concurrent writes would streamline the discussion quite a bit}

% \nc{I honestly would just merge the two descriptions: Start the paragraph saying: processes writes asynchronously, but needs to be able to check integrity, even in the presence of concurrent writes. To do so, it maintains two data structures ... blah blah. If no pending, then broadcast. Otherwise wait (do you already reply before waiting, or only after?)}

Altogether, this mechanism converts concurrent writes to the same block into sequential writes.
This is similar to the approach taken in Harmonia~\cite{harmonia}, CrossFS~\cite{crossFS}, and
dm-integrity~\cite{dmIntegrity}, which represents the state-of-the-art in the understanding of block device semantics.

\subsubsection{Asynchronous writes arriving at the backups}
\label{sec:async-writes-backup}
Once a write arrives at the backups, the backups must determine the order in which to submit the writes to disk.

The na\"ive solution, executing all writes one-after-the-other according to \texttt{writeIndex}, is a non-starter performance-wise.
The challenge is then parallelizing these writes safely.
To do so, the backups need to determine which writes are to the same block, as block semantics allows non-conflicting writes to be ordered arbitrarily~\cite{crashmonkey,crashMonkeyAndACE,dmIntegrityMailingList}.

We make the following observation: the mechanism used by the primary to \emph{avoid} conflicting writes can be reused by the backups to \emph{permit} non-conflicting concurrent writes.

In order to preserve happens-before relationships between writes to the same block, the backups must still submit writes to disk in order of \texttt{writeIndex}$_p$ as assigned by the primary, but do not wait for the disk to finish processing previous writes; only conflicting writes need to block.
% \heidi{can you elaborate on why? if the disk isn't guaranteed to process writes to different blocks in the same order why does sending order matter?}
Concretely, once a backup $b$ receives a write with $\texttt{writeIndex}_p = \texttt{writeIndex}_b + 1$, it atomically increments \texttt{writeIndex}$_b$, then follows the same process depicted in \Cref{fig:pending-queue-invoked-tree}.
% This simultaneously allows disk bandwidth to be fully utilized on the backup, allowing non-conflicting writes to be concurrently in-flight, while preserving write ordering over individual blocks.
This simultaneously allows non-conflicting writes to be concurrently in-flight while preserving write ordering over individual blocks.

\subsubsection{Synchronous writes}
\label{sec:sync-writes}
Writes tagged with persistence flags are handled identically with one exception: \sys{} does not return the write until backups confirm that they have received all writes with a lower \texttt{writeIndex}.

This subsumes the behavior of both persistence flags.
A write tagged with \texttt{REQ\_FUA} simply needs to be recoverable from the backups, which is clear from the acknowledgment.
A write tagged with \texttt{REQ\_PREFLUSH} requires the persistence of all writes that happen-before it (\Cref{def:happens-before}).
By assigning \texttt{writeIndex} based on invocation order, the primary guarantees that if another write happens-before the \texttt{REQ\_PREFLUSH}, it must have a smaller \texttt{writeIndex}.
Therefore, when a backup acknowledges the \texttt{REQ\_PREFLUSH}, it must have already received the earlier write.
% \nc{This is sequential, does it ever become a bottleneck? Do we need to say something about it?} \david{Submitting to disk is sequential but non-blocking, so the disk still processes multiple writes concurrently and the submission is never the bottleneck} \nc{Talkig about the counter speciificially} \david{Oh ok, that's just an atomic int so it's negligible.}
This design forces \texttt{REQ\_FUA} to behave like \texttt{REQ\_PREFLUSH}, which may increase latency as the backup unnecessarily waits for all previous writes to arrive.
This is intentional.
If backups could acknowledge \texttt{REQ\_FUA}s without waiting for all prior messages, then different backups may be ``fresher'' for different blocks.
Two backups may have each received and acknowledged a different \texttt{REQ\_FUA}, and upon failure and recovery, the primary would be unable to select a single freshest backup to recover from.

% \nc{It might be nice to justify why we made this decision}

\subsubsection{Reads}
\label{sec:reads}
Reads are performed on the primary and do not involve the backups.
To maintain integrity for concurrent reads and writes to the same block, \sys{} inserts reads in the same pending queue/involved tree as concurrent writes.
% \nc{My first thought was, why do you have to do this for reads? Can you add some justification?}
Once the read can be executed, \sys{} fetches the corresponding page from disk, decrypts it, checks it against the hash in memory, and returns the decrypted page if the integrity check succeeds.
% \natacha{strange flow between the "To maintain integrity" paragraph and this one "When sys receives a read from a higher layer. Shouldn't we start by the When sys receives a read from a higher layer?}
% If the check fails---because of a rollback attack or a benign disk corruption---\sys{} crashes the machine, entering recovery upon restarting, where the backups are used to verify the integrity of the primary's disk and recover corrupted data.
If the check fails---because of a rollback attack or a benign disk corruption---\sys{} crashes the machine, entering recovery upon restarting.

\subsection{Recovery}
\label{sec:recovery}

\sys{}'s recovery protocol differs from traditional disk recovery in two ways.
First, it must retrieve in-memory integrity metadata and any corrupted disk pages from the most up-to-date backup.
Second, it must prevent split-brain attacks, where an attacker could feign a crash, wait for the user to ``restart'' the ``crashed'' machine while actually starting a new machine, then route external client traffic between the new and ``crashed'' machines as desired~\cite{narrator}.
% \nc{Still can't really follow the example. You have TLS connections how can you have one machine spoof another machine} \david{The idea is that a third party that needs to access the application (such as a database client) that doesn't have the TLS connection could be routed to a stale database}

We prevent split-brain attacks during recovery by drawing an equivalence to \textit{reconfiguration}~\cite{verticalPaxos,paxosComplex}.
We require the client (or some fault-tolerant third party) to provide each restarted machine a new identity, even if the physical hardware is the same.
Each recovering machine then joins a new configuration that excludes its crashed self, ensuring that stale machines no longer participate in the protocol.

\sys{}'s recovery protocol is based on Matchmaker Paxos~\cite{matchmakerPaxos}, a state-of-the-art vertical reconfiguration~\cite{verticalPaxos} protocol that uses two round-trips: one to a fault-tolerant \revisionbarsTwo{third party}{consensus protocol} to establish the current configuration (the active primary and backups), and another to invalidate all previous configurations.
\revisionbarsTwo{
We use CCF~\cite{ccf}, a TEE-based distributed key-value store, as the third party.
Relying on a third party for reconfiguration is a practice common among consensus protocols~\cite{verticalPaxos,chainReplication,boxwood,gfs,pacifica} and does not affect the critical path of \sys{}.

As CCF is only utilized during recovery, a single CCF service could support many \sys{} instances.
}
{

We use CCF~\cite{ccf} as the consensus protocol as opposed to implementing our own.
CCF is a mature, fault-tolerant, rollback-resistant, replicated key-value store that requires no additional trust assumptions.
As CCF is only needed at the beginning of recovery, its performance does not affect \sys{}; it can be configured to run on the same machines as \sys{} or shared across many \sys{} instances to minimize cost.}
% \nc{I think it's much better than it was before, I would maybe just add one or two sentences more about what vertiical reconfiguration actually does. As soujanya said, maybe a brief count of how many phases they are, etc.}

% During recovery, the recovering node identifies the most up-to-date node---the \textit{designated} node with the latest configuration and highest \texttt{writeIndex}---then copies the designated node's integrity metadata, scans its local disk, and copies over any corrupted pages.
During recovery, the recovering node contacts at least $N-f$ nodes from the latest configuration, identifies the node with the highest \texttt{writeIndex}, copies that node's integrity metadata, scans its local disk, and copies over any corrupted pages.
The recovering node then alerts any nodes in its new configuration, which copy their integrity metadata and corrupted pages from the recovering node as well in order to maintain consistency.
Once the disk is repaired, \sys{} is mounted and can be used as-is.

If fewer than $N-f$ nodes from the latest configuration can be reached during recovery, then \sys{} assumes the worst---that its disk has been rolled-back by a Type II attacker---and aborts recovery.

\revisionbarsThree{Additional details}{The full reconfiguration protocol and its proof} are \tr{in \Cref{sec:appendix-recovery}}{in the supplementary material}.

\section{Evaluation}
\label{sec:evaluation}

\sys{} seeks to provide general and automatic rollback resistance with minimal performance overhead.
In this section we answer the following questions:
\begin{enumerate}
    \item Generality and Automatability: Can \sys{} support unmodified applications, and at what cost? (\Cref{sec:performance-overview})
    \item Performance: How does \sys{} compare against non-automatic rollback resistance solutions? (\Cref{sec:nimble})
    \item Performance: How do \sys{}'s overhead vary as a function of the workload? (\Cref{sec:microbenchmarks})
%    \item Correctness: Is \sys{} robust to failure? (\Cref{sec:eval-recovery}) \souj{How robust -- quantifying it seems important? if we may not then we don't list it here.}
%    \nc{I'd not list it as a main question, we already have 3. Doesn't mean we dont't have it, just don't think it needs to be a top section}
\end{enumerate}

\par \textbf{Implementation.} We implemented \sys{} as a device mapper for Linux kernel 6.8 available at \github{} (3,980 LoC).
AEAD uses in-kernel AES-GCM; hashing uses HMAC-SHA256.
We use in-kernel TCP connections with signed messages for primary-backup communication.

\par \textbf{Experimental setup.}
\label{sec:experimental-setup}
We use Azure \texttt{DC16ads\_v5} machines (16 vCPUs, 64GB RAM, 10 Gbps network, AMD SEV-SNP TEE) in the North Europe region.
Ping time is 0.3ms.
We mount \sys{} over local disk to avoid the default replication Azure provides (which does not protect against rollbacks).
Experimental results are the average over 3 runs.
% \natacha{devil's advocate: what does RB do extra that Azure doesn't do? If azure already replicates, why do we need rollbaccine}
% \david{Traditional integrity checking assumes that if you put integrity info on disk, it's highly unlikely that the integrity and data will be corrupted together in a correct way. In Rollbaccine's setting, it's highly likely. We also have the async writes trick.}
% \heidi{+1. Both do replication (three way in the case of managed disks) but managed disk assume the host is trusted and just that disk might occasionally be corrupted, whereas rollbaccine has a much strong threat model that assume that the node might be actively malicious}

We compare \sys{}'s performance against four systems: Unreplicated, DM, Replicated, and Nimble.
\textbf{Unreplicated} reads and writes from local (ephemeral) disk without replication.
It represents the highest-performing but least durable and secure option.
% \textbf{DM} adds dm-crypt and dm-integrity for encryption and integrity validation, using the same AES-GCM cipher as \sys{}.
dm-crypt + dm-integrity provides confidentiality and detection of random data corruptions.
\textbf{DM} adds dm-crypt and dm-integrity for encryption and detection of random data corruptions, using the same AES-GCM cipher as \sys{}.
% \souj{Could we list and cite systems to justify that run in this config?}
\textbf{Replicated} uses the highest-performing durable disk available to Azure VM-based TEEs, a locally 3-way replicated P80 Premium SSD rated for 20,000 IOPS.
Both DM and Replicated write integrity metadata to disk~\cite{azureStorage}, which is not sufficient against rollback attacks; the integrity hash could be rolled back along with the data by a motivated attacker.
% , and \changebars{thus}{is less secure, as} it does not prevent an attacker from rolling back the integrity-protected data along with the integrity hashes. 
% Like DM, \changebars{it}{replication} similarly provides integrity protection against random data corruption; however, the integrity hash could be rollback along with the dat by a motivated attacker.
% \natacha{Hard to parse that last sentence. When could that arise? Conceptually, what's the difference regarding integrity between Replicated and DM}.
% \heidi{@Natacha, I've rewritten it should be clearer}
% None of the Unreplicated, DM, and Replicated settings provide rollback resistance.

\textbf{Nimble}~\cite{nimble} is a state-of-the-art solution against rollback attacks that is general, resistant, but not automatic.
Applications must be \emph{manually modified} to send state updates to a ``coordinator'' that persists the updates to untrusted storage, replicates to 3 TEE-based ``endorsers'', and then replies to the application.
These modifications are labor-intensive; it took three person-months to modify HDFS into NimbleHDFS~\cite{nimble}.

We evaluate against four configurations of NimbleHDFS: \textbf{NimbleHDFS-100}, \textbf{NimbleHDFS-100-Mem}, \textbf{NimbleHDFS-1}, and \textbf{NimbleHDFS-1-Mem}.
The number (100 or 1) represents batch size.
The original paper batches and \revisionbarsThree{replicating}{replicates} every 100 writes, creating a window of vulnerability during which writes marked ``durable'' may be rolled back by an attacker~\cite{nimble}, breaking the semantic guarantees of HDFS.
Setting batch size to 1 preserves semantics.
The -Mem modifier indicates whether state updates are persisted to locally replicated \texttt{Standard LRS} storage as described in the paper or kept in the coordinator's memory (and not fault tolerant).
We co-locate the coordinator machine with the NimbleHDFS to reduce network latency.
% Azure storage uses local replication (\texttt{Standard LRS}) instead of geo-replication to avoid unfair cross-region latency penalties.

\sys{} is evaluated with \revisionbarsTwo{six}{seven} configurations.
\textbf{\sys{}} is the standard setup, with $f=1$ and $L=0$ (all 2.4GB of integrity metadata in memory).
\revisionbarsTwo{}{\textbf{\sys{}-multicloud} uses a GCP \texttt{n2d-standard-16} machine (16 vCPUs, 64GB RAM, AMD SEV-SNP TEE) in the West Europe region as the backup to evaluate the cost of cross-cloud deployments.
Cross-cloud deployments guarantee rollback resistance even in the presence of Type II attackers (malicious cloud providers).
Its ping time to the Azure machine is 23ms.}
\textbf{\sys{}-sync} synchronously replicates all writes regardless of persistence flags in order to isolate the effect of asynchronous replication.
\textbf{\sys{}-f=0} and \textbf{\sys{}-f=2} toggle between no backups (only rollback detecting) and 2 backups, measuring the overhead of networking.
\textbf{\sys{}-L=1} and \textbf{\sys{}-L=2} place the bottommost $L$ layers of the integrity metadata Merkle tree on disk, measuring the overhead of read/write amplification, requiring only 0.15GB and 9.6MB of memory for integrity metadata respectively.

\subsection{Performance Overview}
\label{sec:performance-overview}

% We evaluate the overheads of \sys{} relative to other baselines with the
We evaluate \sys{} with the following benchmarks and unmodified applications:
TPC-C~\cite{tpc-c} over PostgreSQL mounted on ext4 (\Cref{fig:postgres}),
NNThroughputBenchmark~\cite{nnThroughputBenchmark} over HDFS~\cite{hdfs} mounted on ext4 (\Cref{fig:hdfs-throughput}), and
Filebench~\cite{filebench} Varmail and Webserver workloads over ext4 and xfs (\Cref{fig:varmail-throughput,fig:varmail-latency}).

% \changebars{\sys{} makes PostgreSQL, HDFS, and Filebench Webserver rollback-safe with less than $15\%$ throughput and latency overhead.
% Filebench Varmail experiences higher overheads with its high frequency of persistence operations.}{}
% \nc{Cut for space if you want. Fine sentence but you're going to talk about results}

\begin{figure}[t]
    \centering
    \includegraphics[width=0.5\linewidth]{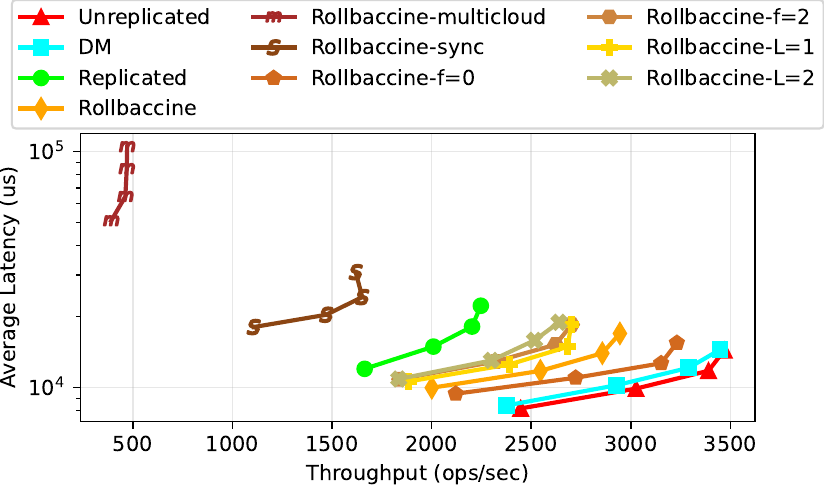}
     \caption{PostgreSQL TPC-C throughput-latency graph\revisionbars{}{ with 20, 30, 40, and 50 clients}.}
     \label{fig:postgres}
\end{figure}

\textbf{PostgreSQL.}
PostgreSQL is a transactional database that guarantees the durability of committed transactions by persisting writes to disk.
Rollback attacks on disk can break durability, allowing attackers to remove unwanted transactions.
% \textit{even if} the transactional database executes within a TEE, allowing attackers to remove unwanted transactions.
PostgreSQL contains 1.3M LoC, making it infeasible to manually rewrite for rollback resistance. It is therefore a prime target for \sys{}, which promises \textit{automatic} rollback resistance.
% \souj{This sentence feels like a misfit here, move it to implementation?} \david{I don't think it would make sense in the implementation, which is about Rollbaccine and not any of the applications or benchmarks}
% We benchmark PostgreSQL with TPC-C, a standard OLTP benchmark for transactional databases.
% We configure TPC-C to run with 10 warehouses and set the isolation level to \texttt{TRANSACTION\_SERIALIZABLE}.
We benchmark PostgreSQL using TPC-C with 10 warehouses and isolation level \texttt{TRANSACTION\_SERIALIZABLE}.
The results are in \Cref{fig:postgres}. \revisionbarsTwo{}{Latency is log-scale.}

Compared to Unreplicated, DM introduces negligible overhead.
Replicated and \sys{} respectively reduce throughput by $35\%$ and $15\%$ and increase latency by $54\%$ and $19\%$.
This can be attributed to the fact that when benchmarked with TPC-C, roughly every 1 in 5 operations in PostgreSQL are persisted, because every transaction must be durably flushed to PostgreSQL's Write Ahead Log (WAL) before commit.
Both Replicated and \sys{} must then synchronously replicate over the network, introducing additional delay, although the latency for Replicated is an order of magnitude greater (\Cref{sec:microbenchmarks}).
Despite this, the performance penalty is not severe because, at 10 warehouses, TPC-C is contention bottlenecked.
% \souj{Traceback the additional overheads from \sys{} relative to replication.}

\revisionbarsTwo{Of the configurations of \sys{}, \sys{}-sync performs the worst, unable to leverage the benefits of asynchronous replication.}
{
Of the configurations of \sys{}, \sys{}-multicloud has the worst performance due to high inter-cloud latency.
\sys{}-sync is also an outlier, as it is unable to leverage the benefits of asynchronous replication.
This represents a lower bound on performance.
The remaining configurations of \sys{} take advantage of the observation that applications, such as PostgreSQL, are \emph{already} designed to minimize persistence and carefully choose when to \texttt{fsync}, so blocking on replication is only necessary for persistent writes.

}
% \nc{Since we have space, I would add some more explanation here. This represents a lower-bound on performance. Applications are *already* designed to avoid synchronous writes and carefully choose when to fsync (due to the high costs of disk). Postgres only issues 20\% of synchronous writes (for instance). Rollbaccine piggy backs on this design principle (or something along those lines).}
The differences between \sys{}-f=0, \sys{} (with f=1), and \sys{}-f=2 illustrate the overhead of networking, whereas the differences between \sys{}, \sys{}-L=1, and \sys{}-L=2 demonstrate the effect of read/write amplification from accessing Merkle tree integrity metadata on disk.

The results confirm that a major component of \sys{}'s high performance stems from its differentiation between synchronous and asynchronous replication, and that \sys{} can switch between different levels of fault tolerance and memory usage without significant penalty.

\begin{figure*}[t]
    \centering
    \includegraphics[width=0.75\linewidth]{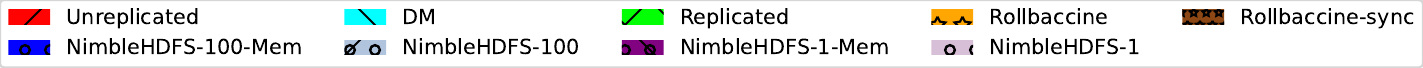}
    \centering
     \begin{subfigure}[b]{0.62\linewidth}
         \centering
         \includegraphics[width=\textwidth]{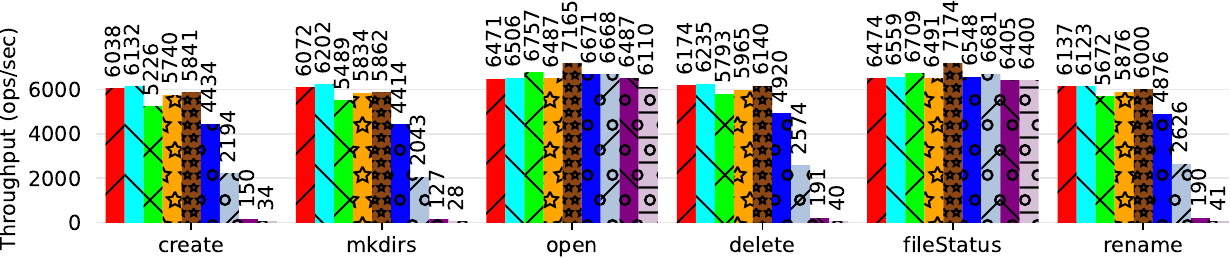}
         \caption{HDFS NNThroughputBenchmark Throughput}
         \label{fig:hdfs-throughput}
     \end{subfigure}
     \hfill
     \begin{subfigure}[b]{0.18\linewidth}
         \centering
         \includegraphics[width=\textwidth]{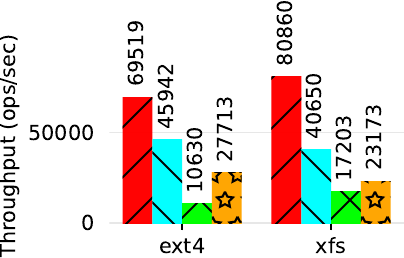}
         \caption{Varmail Tput}
         \label{fig:varmail-throughput}
     \end{subfigure}
     \hfill
     \begin{subfigure}[b]{0.18\linewidth}
         \centering
         \includegraphics[width=\textwidth]{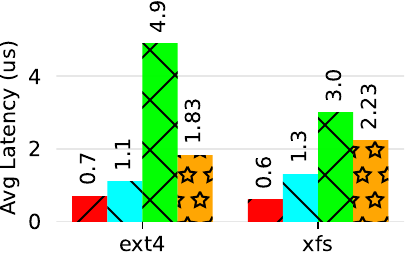}
         \caption{Varmail Latency}
         \label{fig:varmail-latency}
     \end{subfigure}
    \caption{Performance results. \revisionbars{Nimble's configurations contain ``$\circ$''.}{Each bar is labeled with the exact throughput on top.}}
    \label{fig:performance-overview}
\end{figure*}

\textbf{HDFS.}
Hadoop Distributed File System is the file system backing Hadoop MapReduce.
% cloud applications running on Hadoop MapReduce.
Rollback attacks can break the persistence guarantees of HDFS~\cite{hdfs}.
% Securing HDFS against rollback attacks required three person-months of manual effort in Nimble~\cite{nimble} and the modification of 1,689 lines of code; \sys{} promises \textit{automatic} rollback resistance with no code modifications.
% We will compare the performance of the manually modified NimbleHDFS and HDFS over \sys{} in detail in \Cref{sec:nimble}.
% In this section, we first evaluate the overhead of HDFS over \sys{}.
We configure HDFS to run with one namenode and evaluate it with Hadoop's NNThroughputBenchmark~\cite{nnThroughputBenchmark}; each operation uses 500,000 files (or directories for \texttt{mkdirs}) and 16 client threads~\cite{nimble}.
Results are in \Cref{fig:hdfs-throughput}.

DM and \sys{} perform similarly to Unreplicated, reducing throughput by at most $5\%$ and at times outperforming Unreplicated (attributed to experimental noise).
This is because
% (1) HDFS commits writes by appending to a file called \texttt{editlogs} and then flushing updates to disk; writes cannot return until the updates are persisted.
% As a result, updates are batched into sequential writes, and the additional cost of a round trip to the backup for \sys{} is amortized.
NNThroughputBenchmark, regardless of the number of client threads, uses a single thread to communicate with HDFS in order to isolate the overhead of RPC calls~\cite{nnThroughputBenchmark}.
Once enough client threads are launched (16 is enough) on Unreplicated, DM, or \sys{}, this single thread becomes the bottleneck, not HDFS.

% \souj{Are we still discussing HDFS? or Missing a "Varmail and Webserver with filebench" heading?} \david{This is still HDFS}
Replicated suffers a higher $13\%$ throughput overhead; its high latency delays file persistence and reduces throughput.

\textbf{ext4 and xfs.}
ext4 and xfs are file systems in the Linux kernel with traditional POSIX semantics that we mount over \sys{}, providing rollback resistance to \textit{any} TEE application that reads and writes to either file system.
% \souj{Say ext4 and xfs are just two POSIX file systems that we run \sys{} unmodified with.}
% \natacha{Let's make sure to highlight this point in the intro}

We emulate such applications with Filebench using the default Varmail and Webserver profiles.
Varmail is a highly synchronous workload that writes and explicitly calls \texttt{fsync} every 4 operations.
Its results can be found in \Cref{fig:varmail-throughput,fig:varmail-latency}.
Webserver is completely asynchronous, executing reads and occasionally appending to a logfile.
Both workloads are run for the default 60 seconds.

The throughput and latency trends are similar for ext4 and xfs, so we will discuss them together.
We first examine Varmail.
Unlike TPC-C (contention bound with 10 warehouses) and NNThroughputBenchmark (bottlenecked on a single thread), Varmail is bottlenecked on disk, so DM, Replicated, and \sys{} all experience throughput and (inversely proportional) latency degradations due to the high volume of synchronous writes.
Replicated has the highest average latency per operation due to its high fsync latency.
% \souj{discuss average and say tail latency spikes are not high} \david{The benchmark does not provide tail latencies}
\sys{} has the second-highest latency, because it must similarly wait for a network round trip, reducing throughput by $71\%$ and increasing latency by $2.7\times$.
DM does not perform networking but still suffers from synchronously flushing journal entries to disk.

In contrast, all configurations perform similarly for Webserver, which does not require any synchronous operations and mostly performs sequential reads that can be served from prefetched pages.

In summary, except for Varmail, \sys{} adds a maximum of $19\%$ overhead to the Unreplicated baseline across diverse workloads.
The fact that \sys{} is able to provide rollback resistance for all these systems without code modifications demonstrates its versatility and ease-of-use.

\subsection{Comparison against Nimble}
\label{sec:nimble}

\revisionbarsThree{
\sys{} provides rollback resistance for any program mounted over \sys{}'s device mapper.
If an application relies on a cloud service for persistence, then the service itself must be mounted over \sys{}, requiring buy-in from cloud providers.

Nimble does not require cloud provider buy-in and instead detects rollbacks in existing services.
Unlike block devices, cloud services have diverse APIs, so Nimble sacrifices automation for compatibility.
In addition, the amount of integrity metadata that must be preallocated for cloud services is often opaque to the end-user, so Nimble opts to maintain a log of updates instead.
This serves Nimble well for the applications it targets, which can tolerate of vulnerability during which data may be rolled back without detection, allowing Nimble to batch log updates.

Unfortunately, the overheads of sequential log replication resurface when batching is disabled for safety.
In NNThroughputBenchmark's write operations (\texttt{create}, \texttt{mkdirs}, \texttt{open}, \texttt{delete}, \texttt{rename}), \sys{} outperforms NimbleHDFS-1 by $208\times$ (\Cref{fig:hdfs-throughput}).

Azure storage plays a role; both NimbleHDFS-100 and NimbleHDFS-1 underperform their in-memory counterparts by $2-5\times$.

Nimble's performance penalties, however, mainly stem from synchronous, sequential log replication.
% NimbleHDFS operations that use multiple threads to save files in parallel must still sequentially append, sign, and replicate each log entry.
NimbleHDFS's multi-threaded file operations must still sequentially append, sign, and replicate each log entry.
Its use of asymmetric ECDSA-SHA256 signatures allows the log to be publicly verifiable but introduces additional overhead in the critical section.
NimbleHDFS's throughput then becomes a function of its batch size, reducing NimbleHDFS-1's throughput to double-digits.

Reads (\texttt{open} and \texttt{fileStatus}) on the other hand are local, so all systems perform similarly.
}
{
Nimble implements rollback resistance by maintaining a replicated log; applications must be modified in order to append state updates to the log.
Nimble then batches those updates in order to improve performance; the way it batches, however, is incompatible with block device crash consistency.
If the batch size is set to $B$, then replication occurs every $B$ writes; the first $B-1$ writes of each batch, persistent or not, will be returned immediately before they are made rollback resistant.
% \nc{Still confusing I think. There's no reason why batching should cause Nimble to be block device crash consistent. It could just wait to respond until the batch is complete. Should we say that Nimble returns immediately and periodically, every B writes, makes any existing writes rollback resistant? For ex, in consensus, batching doesn't cause the records within the batch to see a weaker consistency guarantee}
NimbleHDFS-1, with a batch size of 1, is therefore the only configuration that is \emph{safe}.

We ask whether a general purpose rollback tolerant solution like \sys{} can match the performance of a manual, rollback-resistant solution.
We compare against NimbleHDFS-1, the only version of Nimble to provide block device crash consistency.
Our results were surprising.
\sys{} not only matches but outperforms NimbleHDFS-1 by $208\times$ on NNThroughputBenchmark's write operations (\texttt{create}, \texttt{mkdirs}, \texttt{delete}, \texttt{rename}) (\Cref{fig:hdfs-throughput}).
The discrepancy holds even when comparing against \sys{}-sync, where we have \sys{} artificially replicate \textit{every} disk write synchronously.

Azure storage plays a role; both NimbleHDFS-100 and NimbleHDFS-1 underperform their in-memory counterparts by $2-5\times$.

But the culprit is CPU (\Cref{sec:performance-overview}).
Round-trip network latency introduces negligible overhead on the CPU-bottlenecked thread (which is why \sys{} and \sys{}-sync perform similarly), whereas the asymmetric ECDSA-SHA256 signatures used by Nimble's messages overwhelm it.
NimbleHDFS's throughput then becomes a function of how many messages it must send; larger batch sizes allow it to amortize signatures across log entries, but once batching is disabled for safety, its throughput reduces to double-digits.

Reads (\texttt{open} and \texttt{fileStatus}) on the other hand are local, so all systems perform similarly.
% Note that the throughput for Nimble in our evaluation is roughly $40\%$ lower than what was originally reported across operations~\cite{nimble}.
% This is because the Nimble paper used compute-optimized \texttt{F64s\_v2} VMs, which do not support TEEs.
% Nevertheless, the relative throughput between the configurations are the same in both papers.
}
%\nc{I'm having a hard time following this paragraph as I think we're over indexing to the reviewer's questions a little bit. It also gets confusing that you spend so long talking about how crazily better we are than nimble for perf, only to conclude it's not fundamental. Can we try to brainstorm a way to reorganise such that the flow is a bit better? My understanding is that the high level points we want to convey are 1) a discussion on the correctness of Nimble relative to Rollbaccine, 2) the fact that we can achieve performance on par with a general solution 3) acknowledge that nimble's model has benefits over Rollbaccine (no cloud buy-in necessary, support for SGX). We could in theory support SGX, but don't. The cloud buy in part is fundamental though.}

% Despite the reduced instance sizes, we were able to replicate the results of NimbleHDFS-100, which still performs $2-3\times$ worse than Unreplicated for the write operations, and is comparable to Unreplicated for read operations .

\subsection{Microbenchmarks}
\label{sec:microbenchmarks}

\begin{figure*}[t]
    \centering
    \hfill
    \includegraphics[width=0.5\textwidth]{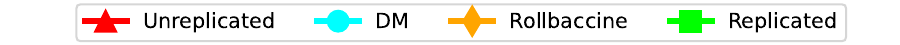}
    \hfill \break
    \centering
    \begin{subfigure}[b]{0.245\textwidth}
         \centering
         \includegraphics[width=\textwidth]{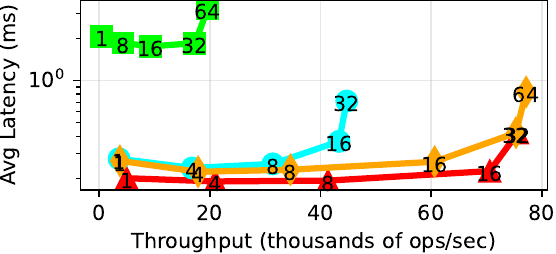}
         \caption{Random read, buffered}
         \label{fig:rand-read-buffered}
     \end{subfigure}
     \hfill
     \begin{subfigure}[b]{0.245\textwidth}
         \centering
         \includegraphics[width=\textwidth]{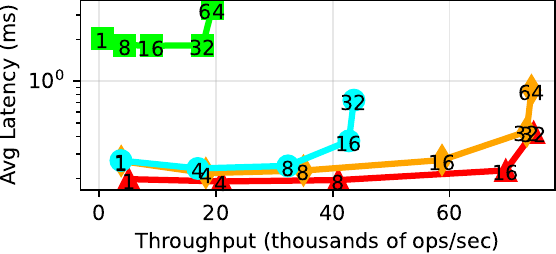}
         \caption{Random read, direct}
         \label{fig:rand-read-direct}
     \end{subfigure}
     \hfill
     \begin{subfigure}[b]{0.245\textwidth}
         \centering
         \includegraphics[width=\textwidth]{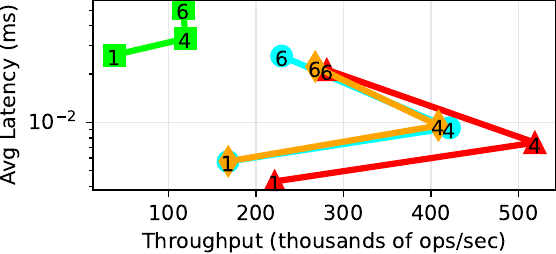}
         \caption{Read, buffered}
         \label{fig:seq-read-buffered}
     \end{subfigure}
     \hfill
     \begin{subfigure}[b]{0.245\textwidth}
         \centering
         \includegraphics[width=\textwidth]{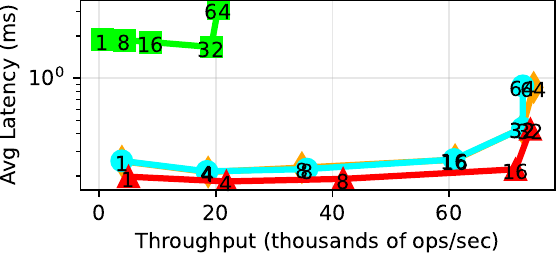}
         \caption{Read, direct}
         \label{fig:seq-read-direct}
     \end{subfigure}
     \hfill
     \begin{subfigure}[b]{0.245\textwidth}
         \centering
         \includegraphics[width=\textwidth]{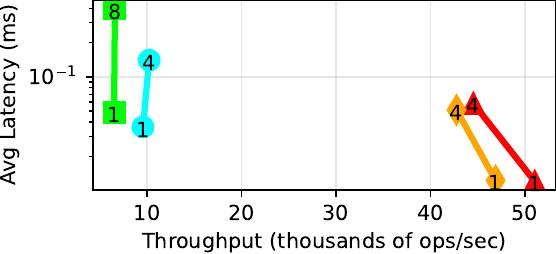}
         \caption{Random write, buffered}
         \label{fig:rand-write-buffered}
     \end{subfigure}
     \hfill
     \begin{subfigure}[b]{0.245\textwidth}
         \centering
         \includegraphics[width=\textwidth]{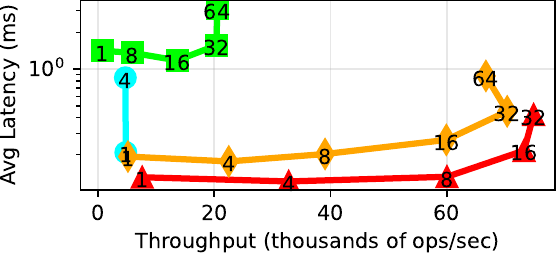}
         \caption{Random write, direct}
         \label{fig:rand-write-direct}
     \end{subfigure}
     \hfill
     \begin{subfigure}[b]{0.245\textwidth}
         \centering
         \includegraphics[width=\textwidth]{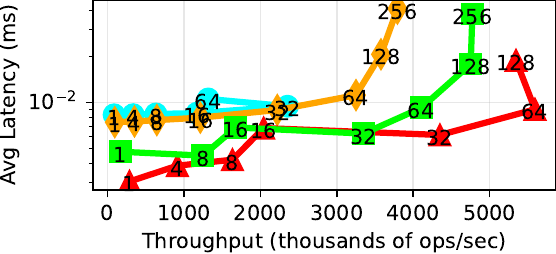}
         \caption{Write, buffered}
         \label{fig:seq-write-buffered}
     \end{subfigure}
     \hfill
     \begin{subfigure}[b]{0.245\textwidth}
         \centering
         \includegraphics[width=\textwidth]{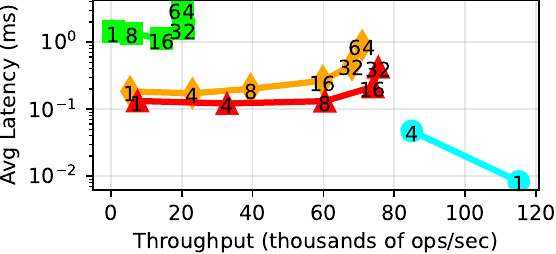}
         \caption{Write, direct}
         \label{fig:seq-write-direct}
     \end{subfigure}
     \hfill
     \begin{subfigure}[t]{0.245\textwidth}
         \centering
         \includegraphics[width=\textwidth]{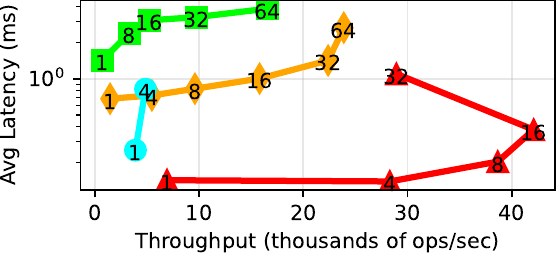}
         \caption{Random write, fsync, buffered}
         \label{fig:rand-write-buffered-fsync}
     \end{subfigure}
     \hfill
     \begin{subfigure}[t]{0.245\textwidth}
         \centering
         \includegraphics[width=\textwidth]{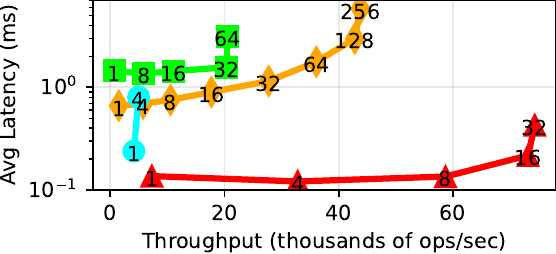}
         \caption{Random write, fsync, direct}
         \label{fig:rand-write-direct-fsync}
     \end{subfigure}
     \hfill
     \begin{subfigure}[t]{0.245\textwidth}
         \centering
         \includegraphics[width=\textwidth]{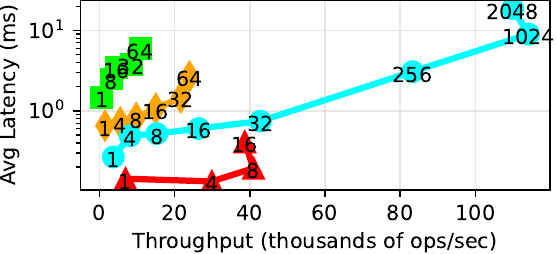}
         \caption{Write, fsync, buffered}
         \label{fig:seq-write-buffered-fsync}
     \end{subfigure}
     \hfill
     \begin{subfigure}[t]{0.245\textwidth}
         \centering
         \includegraphics[width=\textwidth]{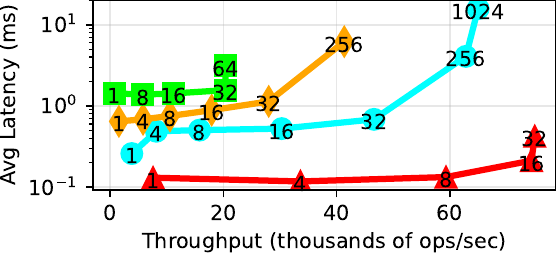}
         \caption{Write fsync, direct}
         \label{fig:seq-write-direct-fsync}
     \end{subfigure}
    \caption{Low contention throughput-latency graphs \revisionbars{of microbenchmarks}{annotated with the number of threads}. Latency is log-scale.}
    \label{fig:microbenchmarks}
\end{figure*}

We analyze \sys{}'s performance with \texttt{fio}, varying I/O direction (read or write), sequentiality (sequential or random), buffering (\texttt{O\_DIRECT} or not), 
\revisionbarsThree{and persistence (synchronous or asynchronous writes)}
{
persistence (synchronous or asynchronous writes), and contention (whether all threads write to the same part of the device or not)}.
All operations are of size 4K with iodepth 1.
We gradually increase the number of fio threads until throughput saturates for each configuration.
For each test, we perform 30 seconds of warmup (filling the page cache), then record statistics for 60 seconds.

Contention is simulated by restricting the set of blocks that each experiment accesses.
In the low contention experiments, all threads start execution from the beginning of the block device, leading to an initial spike in contention that tapers off as slightly faster threads no longer contend with slower ones; in the high contention experiments, threads repeatedly access the same initial 4K bytes.

\Cref{fig:microbenchmarks,fig:contention} display the throughput (thousands of IOPS) and average completion latency (ms) \revisionbarsThree{for each experiment}{of low- and high-contention experiments, respectively}.
Each plot point in the graph is annotated with the number of threads used.
Note that latency \revisionbarsThree{}{(and throughput in \Cref{fig:contention})} is log scale, and that the throughput and latency scales change for each graph.
% For example, in \Cref{fig:seq-write-buffered}, latency is in the $10^{-2}$ms and throughput goes up to 6,000,000 IOPS, while in \Cref{fig:rand-read-direct}, latency is in the $10^{-1}$ms and throughput only goes up to 70,000 IOPS.

We first describe general trends.

\textbf{Direct I/O or persisted writes.}
When either \texttt{O\_DIRECT} or fsync are used for writes, latency increases to the sub-millisecond range and throughput caps at around 75-150,000 IOPS across all tests (\Cref{fig:rand-read-direct,fig:seq-read-direct,fig:rand-write-direct,fig:seq-write-direct,fig:rand-write-direct-fsync,fig:seq-write-direct-fsync,fig:contention})\revisionbarsThree{}{ except for high contention reads over DM, explained below}.
This is because the disk cannot coalesce writes, either because it receives each operation individually (\texttt{O\_DIRECT}) or requires immediate persistence (fsync).

\textbf{Random access.}
Random accesses cap out at 50-80,000 IOPS and sub-millisecond latency (\Cref{fig:rand-read-buffered,fig:rand-read-direct,fig:rand-write-direct,fig:rand-write-direct-fsync,fig:rand-write-buffered-fsync}), with the exception of buffered writes in \Cref{fig:rand-write-buffered}.
For buffered reads, the cap is imposed because page prefetching is ineffective for random accesses and each read must be individually serviced by disk.
For buffered writes, latency is an order of magnitude lower, because although the writes are random, they can still be batched in the page cache and immediately returned.
Throughput, however, is quickly capped once writes fill the page cache and the disk becomes the bottleneck.

\begin{figure}[t]
    \centering
    \begin{subfigure}[b]{0.25\linewidth}
         \centering
         \includegraphics[width=\textwidth]{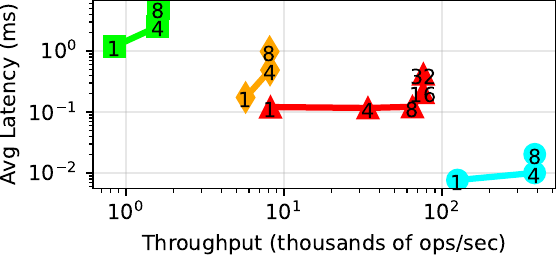}
         \caption{Write, direct}
         \label{fig:write-high-contention}
    \end{subfigure}
    \begin{subfigure}[b]{0.25\linewidth}
         \centering
         \includegraphics[width=\textwidth]{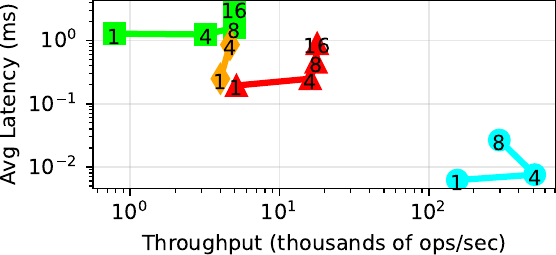}
         \caption{Read, direct}
         \label{fig:read-high-contention}
    \end{subfigure}
    \caption{\revisionbars{}{High-contention graphs similar to \Cref{fig:microbenchmarks}. Both latency and throughput are log-scale.}}
    \label{fig:contention}
\end{figure}

We now explain the performance of each configuration.

\textbf{Unreplicated.}
Reads reach a peak throughput of around 75,000 IOPS and sub-millisecond latency (\Cref{fig:rand-read-buffered,fig:rand-read-direct,fig:seq-read-direct}), except for buffered sequential reads, which reach 500,000 IOPS and $10^{-2}$ms latency (\Cref{fig:seq-read-buffered})\revisionbarsThree{. This is because unlike the other read workloads, buffered sequential reads can consistently read prefetched pages from the page cache.
However, as the number of threads increase, each thread (sequentially) reads from a different location on disk, lowering the efficacy of prefetching and capping throughput.
This behavior is universal across configurations.}{, and high contention reads, which are capped at 18,000 IOPS (\Cref{fig:read-high-contention}).
Buffered sequential reads benefit from prefetching, although as the number of threads increase, each thread (sequentially) reads from a different location on disk, lowering the efficacy of prefetching and capping throughput.
This behavior is universal across configurations.
High contention reads fail to scale across all configurations (except for DM, explained below) due to read collisions at the SSD level~\cite{ssdPerformance,ssdReadCollisions}.
}
% \nc{What do you mean by "suggesting"?}

The throughput and latency of Unreplicated is identical for all write workloads with \texttt{O\_DIRECT} (\Cref{fig:rand-write-direct,fig:seq-write-direct,fig:rand-write-direct-fsync,fig:seq-write-direct-fsync,fig:write-high-contention}), regardless of sequentiality\revisionbarsThree{ or persistence}{, persistence, or contention}, since those writes are disk I/O bottlenecked.
For buffered, persisted writes (\Cref{fig:rand-write-buffered-fsync,fig:seq-write-buffered-fsync}), fsync latency spikes and cripples throughput due to the constant flushing of the page cache.

For the remaining workloads, the behavior of random buffered writes (\Cref{fig:rand-write-buffered}) is explained in the paragraph on random access, and sequential buffered writes (\Cref{fig:seq-write-buffered}) simply measure how quickly full pages can be flushed to disk.

\textbf{DM.}
The majority of overhead for DM comes from dm-integrity~\cite{dmIntegrity}, which maintains a journal of write blocks and their integrity metadata on disk.
The journal entry is flushed to disk when persistence is required, and data is asynchronously copied from the entry to their actual locations on disk.
When a read is requested, if the metadata is not in memory, it must also be fetched from disk.

Fetching metadata is expensive for random accesses, which explains DM's early saturation for random reads (\Cref{fig:rand-read-buffered,fig:rand-read-direct}).
For random writes, the asynchronous copying of data from the journal entry to random regions of disk becomes the throughput bottleneck (\Cref{fig:rand-write-buffered,fig:rand-write-direct,fig:rand-write-buffered-fsync,fig:rand-write-direct-fsync}).

For direct, non-persisted, sequential writes, DM has significantly lower latency than all other configurations (\Cref{fig:seq-write-direct}).
This is because while other configurations directly submit write I/Os to disk, DM builds its own internal cache in the form of asynchronous journal flushes.
Once persistence is required, this no longer gives DM an edge in latency (\Cref{fig:seq-write-direct-fsync}).

% Finally, we must explain how DM's throughput continues rising for sequential, persisted writes (\Cref{fig:seq-write-buffered-fsync,fig:seq-write-direct-fsync}).
% This can again be attributed to journaling.
For sequential, persisted writes (\Cref{fig:seq-write-buffered-fsync,fig:seq-write-direct-fsync}), DM's throughput continues rising due to its journaling.
Although journal entries must be flushed to disk after an fsync, a single journal entry's flush can account for the persistence of multiple writes, in effect batching the fsyncs.

\revisionbarsThree{}{Journaling also allows DM to minimize disk accesses when it comes to high contention (\Cref{fig:contention}).
Reads from journal entries already in-memory can be serviced without going to disk, and writes can be returned immediately after a journal entry is created in-memory.}

\textbf{Replicated.}
Throughput and latency for Replicated is capped by Azure at 20,000 IOPS and millisecond latency, except for sequential buffered reads, random buffered writes, and sequential buffered writes (\Cref{fig:seq-read-buffered,fig:rand-write-buffered,fig:seq-write-buffered}), which benefit from page prefetching and caching.

\textbf{\sys{}.}
Reads in \sys{} perform similarly to Unreplicated (\Cref{fig:rand-read-buffered,fig:rand-read-direct,fig:seq-read-buffered,fig:seq-read-direct}) because they do not leave the primary, with a maximum of $16\%$ and $21\%$ additional latency and throughput overheads as the result of decryption and maintaining the list of invoked and pending operations; the latter happens in a critical section (\Cref{sec:critical-path}).
\revisionbarsThree{}{High contention reads (and writes) are the exception (\Cref{fig:contention}), as the conflicting operations are sequentially executed.}

For asynchronous\revisionbarsThree{}{, low contention} writes, \sys{} scales with the number of threads alongside Unreplicated, with a maximum latency and throughput overhead of $43\%$ and $45\%$ respectively (\Cref{fig:rand-write-buffered,fig:rand-write-direct,fig:seq-write-buffered,fig:seq-write-direct}).
With the exception of sequential buffered writes, which is bottlenecked on bandwidth (\Cref{fig:seq-write-buffered}), the primary's disk is the bottleneck.
These results demonstrate that by replicating asynchronous writes in the background, \sys{} is able to scale.

Persisted writes, on the other hand, are bottlenecked on round-trip time to the backups, with a maximum of $433\%$ and $45\%$ latency and throughput overhead (\Cref{fig:rand-write-buffered-fsync,fig:rand-write-direct-fsync,fig:seq-write-buffered-fsync,fig:seq-write-direct-fsync}).
Latency increases by an order of magnitude as the primary waits for the backup to receive all previous operations before acknowledging the write.
Throughput, however, can continue to scale due to this optimization: if multiple synchronous writes concurrently arrive at the backup, then it only acknowledges the write with the highest index, since that acknowledgment implies the receipt of all prior writes.

% \souj{I would move the summary up until the under-performance is not fundamental upfront.}
In summary, \sys{} adds $21\%$ overhead for reads, $45\%$ overhead for asynchronous writes, similar to DM (with the exception of direct writes), and an order of magnitude of overhead for synchronous writes\revisionbarsThree{}{ and high contention operations}.
% For reads and asynchronous writes, \sys{} offers similar overheads to DM (with the exception of direct writes) while in addition providing rollback resistance.
The under-performance of direct writes\revisionbarsThree{}{ and high contention operations} is not fundamental; \sys{} can be modified to cache \revisionbarsThree{writes in memory as well}{and service reads and writes from memory similar to DM}.
For synchronous writes, \sys{} experiences much higher overheads, but, as seen in \Cref{sec:performance-overview}, most applications are designed to use persistence operations sparingly and are minimally affected.

In addition, \sys{} consistently outperforms Replicated in all benchmarks and microbenchmarks (except sequential buffered writes, which can be cached), suggesting it can be eventually added to Azure storage without a significant performance penalty and provide all applications with rollback resistance by default.

\subsection{Crash Consistency and Recovery}
\label{sec:eval-recovery}

\begin{figure}[t]
    \centering
    \hfill
    \includegraphics[width=0.5\linewidth]{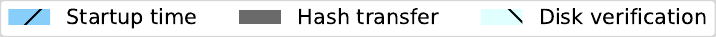}
    \hfill \break
    \centering
    \begin{subfigure}[b]{0.25\linewidth}
         \centering
         \includegraphics[width=\textwidth]{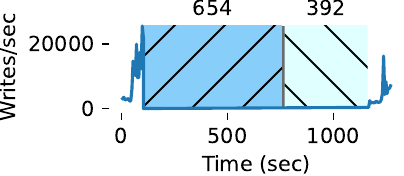}
         \caption{Primary recovery}
         \label{fig:primary-recovery}
    \end{subfigure}
    \begin{subfigure}[b]{0.25\linewidth}
         \centering
         \includegraphics[width=\textwidth]{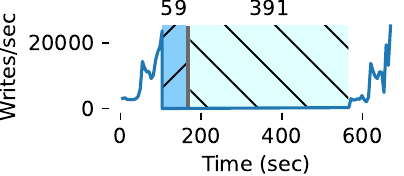}
         \caption{Backup recovery}
         \label{fig:backup-recovery}
    \end{subfigure}
    \caption{Recovery latency.\revisionbars{}{ Each phase is labeled with its latency on top; hash transfer (unlabeled) takes 11 seconds in both experiments.}}
    \label{fig:recovery}
\end{figure}

% So far we have only evaluated the performance of \sys{} when there are no failures.

We simulate rollback attacks on both the primary and the backup in order to analyze recovery latency and correctness.

We first break down the performance impact of recovery in \Cref{fig:recovery}, plotting time against the number of writes processed by the recovering node.
% , since the backup does not process reads.
We start with a standard \sys{} deployment executing PostgreSQL with TPC-C, as in \Cref{sec:performance-overview}.
As it executes, we restart either the primary or backup, overwrite the first 100MB of the 600GB disk to simulate corruption, conduct recovery, then resume TPC-C over the recovered database.
Recovery ends after the last shaded region; the following lull in throughput corresponds to TPC-C setup and is present at the beginning of the graph as well.
The spikes in throughput are a product of the diverse transactions in TPC-C and are unrelated to recovery.

We break the latency of recovery into three main phases in \Cref{fig:recovery}: startup, hash transfer, and disk verification.
Startup time depends on whether Azure physically restarts the machine or redeploys it on a fresh VM; the decision is out of our control.
In our experiment, the primary was physically restarted, and the backup was redeployed, taking 655 and 60 seconds respectively.
Hash transfer is the time it takes for the recovering node to receive the 2.4GB in-memory integrity metadata from the other node; this takes 11 seconds in both tests.
Disk verification is the time it takes for the recovering node to read its the entire disk and perform integrity checks, recovering corrupted pages from the other node when necessary; this takes around 395 seconds in both tests, amounting to 1.5GB/s.
This verification latency is unavoidable for any integrity-preserving application and is comparable to the 600 seconds it takes for dm-crypt + dm-integrity to format the disk.
\revisionbarsOne{}{Recovery time increases linearly with the amount of corrupted disk; 100GB, 300GB, and 600GB of corrupted disk takes an additional 316, 951, and 1941 seconds respectively to recover, around 0.31GB/s.}

We then test the correctness of \sys{} by simulating crashes and verifying the consistency of mounted file systems with ACE~\cite{crashMonkeyAndACE} and xfstests, standard tools for testing crash consistency.
We generate and evaluate 577 tests on ext4 mounted over \sys{}.
\sys{} passes all tests.
\section{Related Work}
\label{sec:related-work}
%We briefly summarize related work.
% \par \textbf{Rollback resistance.}
% Existing application-agnostic solutions for rollback resistance either sacrifice automation or generality.
% Nimble~\cite{nimble} modifies applications to use its API in order to determine when data must be replicated.
% % \changebars{Besides the additional labor required, modifying applications means either manually maintaining two versions or just switching to the modified version, which will have an unnecessary complexity for existing applications even if hidden behind a compiler time flag.}{}
% Narrator~\cite{narrator} assumes that applications are deterministic based on input ordering; accommodating non-determinism requires recording executions with high performance costs~\cite{dOS}.

% \par \textbf{Rollback detection.}
% Solutions that use hashes to verify integrity, but do not keep a backup of the data, are rollback \emph{detecting} but not \emph{resistant}~\cite{vpfs,mlsdisk,dm-x,lcm,rote}.

\par \textbf{Solutions against rollback attacks.}
No existing solution is simultaneously general, automatic, and rollback resistant.
\revisionbarsOne{Nimble~\cite{nimble} sacrifices automation, requiring applications to use its new API for rollback resistance.}{

Unlike \sys{}, which relies on device mappers for automation, Nimble~\cite{nimble} \emph{sacrifices} automation, requiring applications to use its new API for rollback resistance.
What it gains in return are (1) the ability to detect rollback attacks without cloud provider buy-in, and (2) a straightforward implementation over Intel SGX.
\sys{}, in contrast, can (1) only protect against rollback attacks for systems mounted over its device mapper, and (2) requires a subset of the Linux kernel to be a part of the TCB (Trusted Computing Base), which is already part of the assumption for VM-based TEEs but requires significant reimplementation for Intel SGX~\cite{graphene,scone}.

}
Narrator~\cite{narrator} sacrifices \emph{generality}, requiring deterministic execution; accommodating non-determinism requires recording executions with high performance costs~\cite{dOS}.
The remaining solutions sacrifice \emph{resistance}, using hashes \revisionbarsOne{}{or counters} to verify integrity without keeping a backup of the data~\cite{vpfs,mlsdisk,dm-x,lcm,rote,crisp}.
% These solutions often use Merkle trees instead of directly storing hashes on disk in order to detect corruptions of integrity metadata.
% The root of the tree is then either replicated or stored in additional trusted storage.
% ROTE~\cite{rote} similarly replicates a counter instead of a Merkle tree root to track state freshness.

% \heidi{may want to briefly mention: https://www.usenix.org/system/files/conference/usenixsecurity16/sec16_paper_strackx.pdf and https://ieeexplore.ieee.org/document/5958041}

% Rollback resistance or detection has also been manually integrated into consensus protocols~\cite{svr3,engraft,ccf} and databases~\cite{speicher,treaty,fastver}.

\textbf{Device mappers.}
Existing Linux device mappers offer some functionality to enforce confidentiality or integrity of disk.
dm-crypt~\cite{dmCrypt} paired with dm-integrity~\cite{dmIntegrity} or dm-verity~\cite{dmVerity} can provide confidentiality and integrity in the presence of benign, random disk corruptions, but the integrity metadata on disk is vulnerable to rollbacks.
% drbd~\cite{drbd} replicates blocks but cannot detect attacks or compare the freshness of disks.

\textbf{File system semantics.}
Prior work has explored substituting persistent file system operations for fault-tolerant replication~\cite{semanticAwareReplication} outside the context of rollback attacks.
Assise~\cite{assise} uses this strategy for a NVM-backed network file system in order to reduce latency.
SCFS~\cite{scfs} and drbd~\cite{drbd} allow users to toggle between replication schemes to replace disk persistence, while Gaios~\cite{paxosDisk} introduces replacements for file-related system calls that replicate to Paxos state machines.
Blizzard~\cite{blizzard} replicates disk but acknowledges flushes before replication, breaking semantics in order to reduce latency.

% \nc{Isn't there more work on this? Including Blizzard, Rethink the Sync} \david{Rethink the sync is quite different, it models all writes as async, and then prevents the application from outputting until the operation is persisted}

% \textbf{Consensus.} \nc{I think this entire section can go}
% \sys{}'s initialization and critical path largely follow that of traditional consensus protocols~\cite{paxos,raft}.
% In the terminology of Paxos~\cite{paxos}, the primary in \sys{} can be thought of as a proposer, and both primary and backup additionally inhabit the roles of acceptors and state machines.
% Unlike consensus protocols, \sys{} only maintains $f+1$ machines and depends on user intervention for liveness, and has weaker semantics for state machine execution.
% \sys{} also allows non-blocking execution of non-conflicting writes on the backups, similar to Generalized Paxos~\cite{generalizedPaxos}, with additional ordering restrictions in order to accommodate \texttt{REQ\_PREFLUSH}.

\section{Conclusion}
\label{sec:conclusion}

\sys{} provides general, automatic, low-overhead rollback resistance 
% \changebars{in a field where high performance and security are traditionally only achievable through careful code modifications.
% \sys{} achieves this }{}
by marrying the inherent asynchrony and concurrency of disk consistency with fault tolerant replication.
\sys{}'s low overhead and generality leads us to believe that it can be transparently integrated into cloud storage systems with minimal effort.

% \heidi{future work: Could you use RDMA to accelerate replication to backups?} \david{Don't think we have space here}
\section*{Acknowledgements}
This work would not be possible without support from Amaury Chamayou, Eddy Ashton, Joe Hellerstein, Ittai Abraham, Mic Bowman, Michael Steiner, Bruno Vavala, Vijay Chidambaram, Srinath Setty, Shadaj Laddad, Dimitra Giantsidi, Alex Miller, Tyler Hou, and Darya Kaviani.
This work was supported by gifts from AMD, Anyscale, Google, IBM, Intel, Microsoft, Mohamed Bin Zayed University of Artificial Intelligence, Samsung SDS, Uber, VMware, and Ripple.

\bibliographystyle{plainurl}
\bibliography{refs}

\tr{
\newpage

\appendix
\section{Recovery}
\label{sec:appendix-recovery}

Before intercepting operations on the critical path, we must ensure that the primary and backups are all executing within TEEs, communicating with each other over secure channels, and cannot be impersonated by a malicious third party.

\textbf{Initialization.}
Initialization achieves these goals through remote attestation and TLS channels.
After the primary and backups perform attestation, they are given the secret key for encryption and the addresses and roles of each member, which they use to establish secure channels and begin execution.
The process becomes complex once recovery is taken into consideration.

\textbf{Recovery protocol.}
Our recovery protocol is based on the reconfiguration protocol from Matchmaker Paxos~\cite{matchmakerPaxos}, with CCF~\cite{ccf} tracking configurations as the matchmakers.

To track configurations, each node maintains a \texttt{seenBallot}, representing the latest configuration it has seen, and a \texttt{ballot}, representing the latest configuration it has been a member in.
Each protocol message must be tagged with the \texttt{ballot} or \texttt{seenBallot} field of the sender, and recipients only accepts messages if their local ballot is no fresher than the messages' ballot.
Intuitively, this means that nodes do not process requests from stale configurations.

We first modify the initialization protocol so that the initial configuration is committed to CCF.
After attestation, nodes are given a \texttt{seenBallot} representing their configuration \texttt{conf}.
The primary then sends \texttt{MatchA<seenBallot$_p$, conf>} to CCF.
CCF adds the configuration to \texttt{allConf} and responds with \texttt{MatchB<ballot$_c$, allConf>}, where \texttt{ballot$_c$} is the highest ballot observed by CCF.
Upon receiving \texttt{MatchB}, the primary checks if $\texttt{ballot}_c = \texttt{seenBallot}_p$ and if $\texttt{allConf} = \{\texttt{conf}\}$; if so, it sets \texttt{ballot}$_p$ to \texttt{seenBallot}$_p$ and can begin intercepting reads and writes.

A recovering node (including backups) follows the same process but will receive at least one prior configuration.
It then preempts all nodes from prior configurations in \texttt{allConf} by broadcasting \texttt{P1a<seenBallot$_i$>}.

Upon receiving \texttt{P1a<seenBallot$_i$>}, each node $j$ sets its $\texttt{seenBallot}_j$ to $\texttt{seenBallot}_i$ if $\texttt{seenBallot}_i$ is larger, then attempts to aid recovery by responding with \texttt{P1b<seenBallot$_j$, ballot$_j$, hashes$_j$, disk$_j$, writeIndex$_j$>}, where \texttt{hashes$_j$} are its in-memory hashes and \texttt{disk}$_j$ is its disk.
The recovering node ignores any \texttt{P1b}s where $\texttt{seenBallot}_i \ne \texttt{seenBallot}_j$.

After receiving at least 1 \texttt{P1b} from each configuration, the recovering node knows that no prior configuration can make progress and now selects the \textit{designated} node to recover its state from.
The designated node $d$ is the node with the highest $(\texttt{ballot}, \texttt{writeIndex})$ pair, ordered lexicographically.
The recovering node replaces its disk with \texttt{disk}$_d$ and sets its hashes to \texttt{hash}$_d$.
It then uses that hash and disk to update other nodes in its new configuration \texttt{conf}$_i$ by sending \texttt{Reconfig<seenBallot$_i$, hashes$_d$, disk$_d$, writeIndex$_d$>}; those nodes replace their own hashes and disks similarly.

We optimize reconfiguration by omitting \texttt{hashes} from \texttt{P1b} and \texttt{disk} from both \texttt{P1b} and \texttt{Reconfig}, only requesting them when necessary.
The recovering node first requests \texttt{hashes}$_d$ only from the designated node.
It then performs an integrity scan over its local disk using \texttt{hashes}$_d$, and only if any individual pages do not pass the integrity check, requests the page from the designated node.
If the designated node fails during this process, the \texttt{hashes} are requested from another designated node (there must be another, since there are at most $f$ failures and each configuration has \changebars{}{at least }\heidi{changebars as paper relaxes this to N nodes}$f+1$ nodes), and the integrity scan is restarted with the new hashes.
The recovering node then sends \texttt{Reconfig} to the other nodes in \texttt{conf}$_i$, which also perform integrity scans and request corrupted pages from the recovering node.
If the designated node is also a node in \texttt{conf}$_i$, then it does not need to process \texttt{Reconfig}.
This is the case for any recovery that replaces a single crashed node.

Any node that completes disk synchronization then sets its \texttt{ballot} to \texttt{seenBallot}$_i$, \texttt{writeIndex} to \texttt{writeIndex}$_d$, \texttt{hashes} to \texttt{hashes}$_d$, and can resume operation.

Once recovery is complete, old configurations can be removed from \texttt{allConf} in CCF through a garbage collection protocol~\cite{matchmakerPaxos} and safely shut down.
\section{Correctness}
\label{sec:correctness}

We provide a proof sketch for the following theorem:
\begin{theorem}
\label{theorem:rollbaccine-is-correct}
All histories produced by \sys{} are block device crash consistent.
\end{theorem}

We must first map the behaviors of \sys{} to the terms used by block device crash consistency.
A node in \sys{} is \textit{active} if $\texttt{ballot}_p = \texttt{seenBallot}_p$; only active primaries can process read and write messages from the application.
A crash $C$ is any period of time during which there is no active primary; this encompasses failures due to integrity violations detected by \sys{}, signaling a rollback attack.
An invocation $O_{inv}$ is any read or write intercepted by the active primary, and a response $O_{res}$ is any response to invocations returned by the active primary to the upper layer.

Note how the definitions of invocation and responses differ from their definitions in block device crash consistency, which define those operations over the block device (instead of the active primary of \sys{}).
The active primary in \sys{} acts as an additional layer between the application and the block device, delaying invocations to the block device to prevent concurrent accesses to the same blocks (\Cref{sec:async-writes-primary}), removing read responses with that fail integrity checks (\Cref{sec:reads}), and synchronous write responses until they are replicated (\Cref{sec:sync-writes}).

% We will construct the proof sketch by establishing the following lemmas.
% (1) In the absence of crashes, the history produced by the primary's disk is block device crash consistent.
% (2) In the absence of crashes, the history produced by the backup's disk is a durable cut of the history on the primary.
% (3) After a crash-then-reconfiguration, the primary's disk recovers to the disk of either the primary or backup of the last configuration where reconfiguration completed.

We start by establishing that the active primary of \sys{} produces linearizable histories in the absence of crashes.

\begin{lemma}
\label{lemma:primary-consistent}
Given an encrypted disk, a crash-free durable cut $\mathcal{D}$ representing its disk state, its corresponding \texttt{hashes}, and a subsequent era $\mathcal{E}$ produced by the primary, the combined history $\mathcal{D}\mathcal{E}$ is linearizable.
\end{lemma}
\begin{proof}
To prove that $\mathcal{D}\mathcal{E}$ is linearizable, we must construct a sequential history $\mathcal{S}$ that respects reads-see-writes, is equivalent to some $\mathcal{E}' \in \textit{trunc}(\textit{compl}(\mathcal{D}\mathcal{E}))$, and contains a superset of the happens-before relationships in $\mathcal{D}\mathcal{E}$.

We create $\mathcal{S}$ by (1) removing pending invocations in $\mathcal{E}$, and (2) creating abstract threads to isolate accesses to each block (creating $\mathcal{E}'$), then (3) shifting responses earlier in each thread such that matching responses immediately follow each invocation.

$\mathcal{S}$ is sequential by construction.

We know that \sys{} processes operations over the same block sequentially based on invocation order (\Cref{sec:async-writes-primary}), which is unchanged in $\mathcal{S}$.
This means that each read must see the previous write, even if the read invocation precedes the write response.
This holds despite rollback attacks, because \sys{} enforces integrity checks for reads (which would otherwise fail).
Since responses immediately follow each invocation in $\mathcal{S}$, each write-read invocation pair satisfies the reads-see-writes precondition and indeed returns the value of the previous write.
Therefore $\mathcal{S}$ respects reads-see-writes.

We know that for all threads $t$, $\mathcal{E}'[t] = \mathcal{S}[t]$ by construction.

We also know that $\mathcal{S}$ preserves all happens-before relationships, because responses were moved earlier (so any invocation that happens-after a response still happens-after it).

By definition, $\mathcal{D}\mathcal{E}$ is linearizable.
\end{proof}

Under the same circumstances, each active backup produces a durable cut of the era produced by the active primary.

\begin{lemma}
\label{lemma:backup-consistent}
Given an encrypted disk, a durable cut $\mathcal{D}$ representing its disk state, its corresponding \texttt{hashes}, and subsequent eras $\mathcal{E}_1, \mathcal{E}_2$ produced by the primary and a backup respectively, $\mathcal{D}\mathcal{E}_2$ is a durable cut of $\mathcal{D}\mathcal{E}_1$.
\end{lemma}
\begin{proof}
We first show that the backups respects any happens-before relationships on the primary.
Writes are assigned \texttt{writeIndex} by the active primary based on invocation order.
By definition of happens-before, $V_1 \prec V_2$ is only possible if $V_1$ precedes $V_2$, which implies that \texttt{writeIndex} of $V_1$ is also less than \texttt{writeIndex} of $V_2$.
Therefore, if a backup submitted $V_2$ to disk, it must have already submitted $V_1$; formally, $V_2 \in \mathcal{E}_2$ implies $V_1 \in \mathcal{E}_2$.

We now show that the backups must contain all completed synchronous writes.
The primary does not return synchronous writes to the application until the backups acknowledge that they have received that write and all prior writes with lower \texttt{writeIndex}es.
Formally, $W_{res}(b,val,sync) \in \mathcal{E}_1$ implies $W_{res}(b,val,sync) \in \mathcal{E}_2$ if $sync$ contains \texttt{REQ\_FUA} or \texttt{REQ\_PREFLUSH}.
By the definition of durable cut, $\mathcal{E}_2$ is a durable cut of $\mathcal{E}_1$, therefore $\mathcal{D}\mathcal{E}_2$ is a durable cut of $\mathcal{D}\mathcal{E}_1$.
\end{proof}

After reconfiguration, the current active primary contains either the disk of the previous active primary or a previous active backup.
The \textit{current active} primary is the one with the highest \texttt{ballot} and \texttt{writeIndex}; a \textit{previous} active primary is one that was current before reconfiguration.
A current or previous active backup is a backup with a \texttt{ballot} matching the current or previous active primary.

\begin{lemma}
\label{lemma:recover-to-primary-or-backup}
During reconfiguration, the current primary or backup must recover the disk state and \texttt{hashes} of either the previous active primary or its backups.
\end{lemma}
\begin{proof}
Reconfiguration follows the protocol of Matchmaker Paxos~\cite{matchmakerPaxos}.
The proof can be derived from that of Matchmaker Paxos; we provide its intuition here.

We prove inductively on the difference between \texttt{ballot}$_x$ on the current primary $x$ and \texttt{ballot}$_y$ on the previous active primary $y$.
Primary $y$ could have only become active by either completing initialization or reconfiguration by sending \texttt{MatchA} to CCF and adding \texttt{conf}$_y$ to \texttt{allConf}.

In the base case, if $\texttt{ballot}_x = \texttt{ballot}_y + 1$, then when primary $x$ sends \texttt{MatchA} to CCF and receives \texttt{allConf} in \texttt{MatchB}, then $\texttt{conf}_y$ must be the highest-\texttt{ballot} configuration in \texttt{allConf}.
Primary $x$ (and its backups) must synchronize their disks and \texttt{hashes} from either primary $y$ or its backups.

In the inductive case, $\texttt{ballot}_x = \texttt{ballot}_y + i + 1$. 
Because there has been no active primaries since the configuration associated with \texttt{ballot}$_y$, no writes could have been made to disk, and each primary and backup must have synchronized their disks from either primary $y$, its backups, or some machine with state equivalent to those machines, and so must primary $x$.
\end{proof}

Combined, the lemmas state that:
at initialization, the current active primary's disk state is linearizable (\Cref{lemma:primary-consistent}), so prior to any crashes, the primary's disk is also block device crash consistent.
Using induction on crashes, we assume that the $i$-th active primary's disk is block device crash consistent.
In the inductive case, after $i+1$ crashes, the current active primary must recover to either the disk of the $i$-th previous active primary or its backups (\Cref{lemma:recover-to-primary-or-backup}), whose disks are durable cuts (\Cref{lemma:backup-consistent}) of the $i$-th primary's, which is still block device crash consistent by the induction hypothesis.
Therefore, whether the primary recovers from the history or its durable cut, it will still produce a linearizable history (\Cref{lemma:primary-consistent}).
By definition, all histories produced by \sys{} must be block device crash consistent (\Cref{theorem:rollbaccine-is-correct}).
}{}

\tr{}{
% Articles V4mod001-V4mod080 use
\received{July 2025}
\received[revised]{October 2025}
\received[accepted]{November 2025}
}

\end{document}